\documentclass{article}%
\usepackage{amsmath,amsthm}%
\usepackage{amsfonts, color}%
\usepackage{amssymb}%
\usepackage{graphicx}
\usepackage{verbatim} 
\usepackage[T1]{fontenc} 
\usepackage[utf8]{inputenc} 
\usepackage{authblk} 
\usepackage{booktabs}
\newtheorem{theorem}{Theorem}

\newtheorem{definition}{Definition}

\newtheorem{lemma}{Lemma}

\newtheorem{proposition}{Proposition}

\usepackage[margin=1in]{geometry}

\newcommand{\bdsm}{\boldsymbol}

\begin{document}

\title{On linear regression for interval-valued data in $\mathcal{K}_{\mathcal{C}}\left(\mathbb{R}\right)$}
\author[1]{Yan Sun\thanks{yan.sun@usu.edu}} 
\author[2]{Chunyang Li\thanks{chunyang.li@aggiemail.usu.edu}} 
\affil[1,2]{Department of Mathematics $\&$ Statistics\\  
Utah State University\\
3900 Old Main Hill\\
Logan, Utah 84322-3900}

\date{}
\maketitle

\begin{abstract}
It has been some time since interval-valued linear regression was investigated. In this paper, we focus on linear regression for interval-valued data within the framework of random sets. The model we propose generalizes a series of existing models. We establish important properties of the model in the space of compact convex subsets of $\mathbb{R}$, analogous to those for the classical linear regression. Furthermore, we carry out theoretical investigations into the least squares estimation that is widely used in the literature. A simulation study is presented that supports our theorems. Finally, an application to a climate data set is provided to demonstrate the applicability of our model. 
\end{abstract}

%\begin{keywords}
%linear regression; random interval; metric space; coefficient of determination; least squares; asymptotic unbiasedness 
%\end{keywords}

%======================================================Introduction========================================================%
\section{Introduction}\label{intro}
Linear regression for interval-valued data has been attracting increasing interests among researchers. See \cite{Diamond90}, \cite{Korner98}, \cite{Gil02,Gil07}, \cite{Manski02}, \cite{Carvalho04}, \cite{Billard07}, \cite{G-R07}, \cite{Neto08,Neto10}, \cite{Blanco11}, \cite{Cattaneo12}, for a partial list of references. However, issues such as interpretability and computational feasibility still remain. Especially, a commonly accepted mathematical foundation is largely underdeveloped, compared to its demand of applications. By proposing our new model, we continue to build up the theoretical framework that deeply understands the existing models and facilitates future developments.   

In the statistics literature, the interval-valued data analysis is most often studied under the framework of random sets, which includes random intervals as the special (one-dimensional) case. The probability-based theory for random sets has developed since the publication of the seminal book of \cite{Matheron75}. See \cite{Molchanov05} for a relatively complete monograph. To facilitate the presentation of our results, we briefly introduce the basic notations and definitions in the random set theory. Let $(\Omega,\mathcal{L},P)$ be a probability space. Denote by $\mathcal{K}\left(\mathbb{R}^d\right)$ or $\mathcal{K}$ the collection of all non-empty compact subsets of $\mathbb{R}^d$. In the space $\mathcal{K}$, a linear structure is defined by Minkowski addition and scalar multiplication, i.e.,
\begin{equation*}
  A+B=\left\{a+b: a\in A, b\in B\right\}\ \ \ \ \lambda A=\left\{\lambda a: a\in A\right\},
\end{equation*}
$\forall A, B\in\mathcal{K}$ and $\lambda\in\mathbb{R}$. A natural metric for the space $\mathcal{K}$ is the Hausdorff metric $\rho_H$, which is defined as
\begin{equation*}
  \rho_H\left(A,B\right)=\max\left(\sup\limits_{a\in A}\rho\left(a,B\right), \sup\limits_{b\in B}\rho\left(b,A\right)\right),\ \forall A,B\in\mathcal{K},
\end{equation*}
where $\rho$ denotes the Euclidean metric.
A random compact set is a Borel measurable function $A: \Omega\rightarrow\mathcal{K}$, $\mathcal{K}$ being equipped with the Borel $\sigma$-algebra induced by the Hausdorff metric. For each $X\in\mathcal{K}\left(\mathbb{R}^d\right)$, the function defined on the unit sphere $S^{d-1}$:
\begin{equation*}
  s_X\left(u\right)=\sup_{x\in X}\left<u, x\right>,\ \ \forall u\in S^{d-1}
\end{equation*}
is called the support function of X. If $A(\omega)$ is convex almost surely, then $A$ is called a random compact convex set. (See \cite{Molchanov05}, p.21, p.102.) The collection of all compact convex subsets of $\mathbb{R}^d$ is denoted by $\mathcal{K}_{\mathcal{C}}\left(\mathbb{R}^d\right)$ or $\mathcal{K}_{\mathcal{C}}$. When $d=1$, the corresponding $\mathcal{K}_{\mathcal{C}}$ contains all the non-empty bounded closed intervals in $\mathbb{R}$. A measurable function $X: \Omega\rightarrow\mathcal{K}_\mathcal{C}\left(\mathbb{R}\right)$ is called a random interval. Much of the random sets theory has focused on compact convex sets. Let $\mathcal{S}$ be the space of support functions of all non-empty compact convex subsets in $\mathcal{K}_{\mathcal{C}}$. Then, $\mathcal{S}$ is a Banach space equipped with the $L_2$ metric 
\begin{equation*}
  \|s_X(u)\|_2=\left[d\int_{S^{d-1}}|s_X(u)|^2\mu\left(\mathrm{d}u\right)\right]^{\frac{1}{2}},
\end{equation*}
where $\mu$ is the normalized Lebesgue measure on $S^{d-1}$. According to the embedding theorems (see \cite{Radstrom52}, \cite{Hormander54}), $\mathcal{K}_{\mathcal{C}}$ can be embedded isometrically into the Banach space $C(S)$ of continuous functions on $S^{d-1}$, and $\mathcal{S}$ is the image of $\mathcal{K}_\mathcal{C}$ into $C(S)$. Therefore, $\delta\left(X, Y\right):=\|s_X-s_Y\|_2$, $\forall X, Y\in\mathcal{K}_\mathcal{C}$, defines a metric on $\mathcal{K}_\mathcal{C}$. Particularly, let 
$$X=[\underline{X}, \overline{X}]=[X^c-X^r, X^c+X^r]$$
be an bounded closed interval with center $X^c$ and radius $X^r$, or lower bound $\underline{X}$ and upper bound $\overline{X}$, respectively. Then, the $\delta$-metric of $X$ is
\begin{equation*}
  \|X\|_2=\|s_{X}(u)\|_2=\frac{1}{2}\left(\underline{X}^2+\overline{X}^2\right)=\left(X^c\right)^2+\left(X^r\right)^2,
\end{equation*}
and the $\delta$-distance between two intervals $X$ and $Y$ is 
\begin{eqnarray*}
  \delta\left(X, Y\right)&=&\left[\frac{1}{2}\left(\underline{X}-\underline{Y}\right)^2+\frac{1}{2}\left(\overline{X}-\overline{Y}\right)^2\right]^{\frac{1}{2}}\\
  &=&\left[\left(X^c-Y^c\right)^2+\left(X^r-Y^r\right)^2\right]^{\frac{1}{2}}.
\end{eqnarray*}

Existing literature on linear regression for interval-valued data mainly falls into two categories. In the first, separate linear regression models are fitted to the center and range (or the lower and upper bounds), respectively, treating the intervals essentially as bivariate vectors. Examples belonging to this category include the center method by \cite{Billard00}, the MinMax method by \cite{Billard02}, the (constrained) center and range method by \cite{Neto08,Neto10}, and the model M by \cite{Blanco11}. These methods aim at building up model flexibility and predicting capability, but without taking the interval as a whole. Consequently, their geometric interpretations are prone to different degrees of ambiguity. Take the constrained center and range method (CCRM) for example. Adopting the notations in  \cite{Neto10}, it is specified as 
\begin{eqnarray*}
  Y_i^c&=&\beta_0^c+\beta_1^cX_{i}^c+\epsilon_i^c,\\
  Y_i^r&=&\beta_0^r+\beta_1^rX_{i}^r+\epsilon_i^r,
\end{eqnarray*}
where $E(\epsilon_i^c)=E(\epsilon_i^r)=0$ and $\beta_0^r, \beta_1^r\geq 0$. It follows that 
\begin{eqnarray*}
  \left(\hat{Y}_i^c-\hat{Y}_j^c\right)^2+\left(\hat{Y}_i^r-\hat{Y}_j^r\right)^2
  =\left[\beta_1^c\left(X_i^c-X_j^c\right)\right]^2+\left[\beta_1^r\left(X_i^r-X_j^r\right)\right]^2.
\end{eqnarray*}
%Consider for example that $\beta_1^c=1$ and $\beta_1^r=2$. Let 
%\begin{eqnarray*}
%  &&\left(X_1^c-X_2^c\right)^2=3,\ \left(X_1^r-X_2^r\right)^2=1;\\ 
%  &&\left(X_1^c-X_3^c\right)^2=1,\ \left(X_1^r-X_3^r\right)^2=2.
%\end{eqnarray*}
%Then,
%\begin{eqnarray*}
%  \left(\hat{Y}_1^c-\hat{Y}_2^c\right)^2+\left(\hat{Y}_1^r-\hat{Y}_2^r\right)^2=7,\\
%  \left(\hat{Y}_1^c-\hat{Y}_3^c\right)^2+\left(\hat{Y}_1^r-\hat{Y}_3^r\right)^2=9.
%\end{eqnarray*}
Because $\beta_0^r\neq\beta_1^r$ in general, a constant change in $\|X\|_2$ does not result in a constant change in $\|Y\|_2$. In fact, a constant change in any metric of $X$ as an interval does not lead to a constant change in the same metric of $Y$. This essentially means that the model is not linear in intervals.  

In the second category, special care is given to the fact that the interval is a non-separable geometric unit, and their linear relationship is studied in the framework of random sets. Investigation in this category began with \cite{Diamond90} developing a least squares fitting of compact set-valued data and considering the interval-valued input and output as a special case. Precisely, he gave analytical solutions to the real-valued numbers $a$ and $b$ under different circumstances such that $\delta\left(Y,aX+b\right)$ is minimized on the data. The pioneer idea of \cite{Diamond90} was further studied in \cite{Gil01,Gil02}, where the $\delta$-metric was extended to a more general metric called $W$-metric originally proposed by \cite{Korner98}. The advantage of the $W$-metric lies in the flexibility to assign weights to the radius and midpoints in calculating the distance between intervals. So far the literature had been focusing on finding the affine transformation $Y=aX+b$ that best fits the data, but the data are not assumed to fulfill such a transformation. A probabilistic model along this direction kept missing until \cite{Gil07}, and simultaneously \cite{G-R07}, proposed the same simple linear regression model for the first time. The model essentially takes on the form of 
%\begin{equation}\label{mod-Gil}
%  Y_i=aX_i+\tilde{\epsilon}_{i},
%\end{equation}  
%where $a\in\mathbb{R}$ and $\tilde{\epsilon}_i$ is a random interval with a fixed expectation $E\left(\tilde{\epsilon}_i\right)=B\in\mathcal{K}_{\mathcal{C}}(\mathbb{R})$. The centered version of (\ref{mod-Gil}) yields 
\begin{equation}\label{mod-Gil}
  Y_i=aX_i+b+{\epsilon}_i,
\end{equation}
with $a,b\in\mathbb{R}$ and $E({\epsilon}_i)=[-c,c], c\in\mathbb{R}$. This can be written equivalently as
\begin{eqnarray*}
  Y_i^c&=&aX_i^c+b+\epsilon_i^c,\\
	Y_i^r&=&|a|X_i^r+c+\epsilon_i^r.
\end{eqnarray*} 
It leads to the following equation that clearly shows linearity in $\mathcal{K}_{\mathcal{C}}$:
\begin{eqnarray*}
  \delta\left(\hat{Y}_i, \hat{Y}_j\right)=|a|\delta\left(X_i, X_j\right). 
\end{eqnarray*}

Some advances have been made regarding this model and the associated estimators. \cite{Gil07} derived least squares estimators for the model parameters and examined them from a theoretical perspective. \cite{G-R07} established a test of linear independence for interval-valued data. However, many problems still remain open such as biases and asymptotic distributions, as anticipated in \cite{Gil07}. This paper presents a continuous development addressing some issues and open problems in the direction of model (\ref{mod-Gil}). First, we relax the restriction of model (\ref{mod-Gil}) that the Hukuhara difference $Y\ominus \left(aX+b\right)$ must exist (see \cite{Hukuhara67}) and generalize the univariate model to the multiple case. We also give analytical least squares (LS) solutions to the model parameters. Second, we show that our model and LS estimation together accommodate a decomposition of the sums of squares in $\mathcal{K}_{\mathcal{C}}$ analogous to that of the classical linear regression. Third, we derive explicit formulas of the LS estimates for the univariate model, which exist with probability going to one. The LS estimates are further shown to be asymptotically unbiased. A simulation study is carried out to validate our theoretical findings, as well as compare our model to CCRM. Finally, we apply our model to a climate data set to illustrate the applicability of our model. 

The rest of the paper is organized as follows: Section 2 formally introduces our model and the associated LS estimators. Then, the sums of squares and coefficient of determination in $\mathcal{K}_{\mathcal{C}}$ are defined and discussed. Section 3 presents the theoretical properties of the LS estimates for the univariate model. The simulation study is reported in Section 4, and the real data application is presented in Section 5. We give concluding remarks in Section 6. Technical proofs and useful lemmas are deferred to the Appendices. 

%======================================================The Model====================================================%
\section{The proposed model}
%===================================================================================================================%
\subsection{Model specification}
We consider an extension of model (\ref{mod-Gil}) to the form
\begin{eqnarray}
  &&\delta\left(Y_i, aX_i+b\right)=\left\|\epsilon_i\right\|_2,\label{mod-1}
  %&&\epsilon_{i}=[\lambda_{i}-\delta_{i},\lambda_{i}+\delta_{i}], \label{mod-2}
\end{eqnarray}
where $E[\epsilon_i]=[-c,c]$, $c>0$. It is equivalently expressed as  
\begin{equation}\label{mod-cases}
  \begin{cases}
   & Y_i=aX_i+b+\epsilon_i,\ \text{if}\ Y_i\ominus \left(aX_i+b\right)\ \text{exists};\\ 
   & Y_i+\epsilon_i=aX_i+b,\ \text{if otherwise}\ \left(aX_i+b\right)\ominus Y_i\ \text{exists}.
	\end{cases}
\end{equation}
This leads to the following center-radius specification
\begin{eqnarray*}
  Y_i^c &=& aX_i^c+b\pm\epsilon_i^c,\\
  Y_i^r &=& \left|a\right|X_i^r\pm\epsilon_i^r,
\end{eqnarray*}
where $E(\epsilon_i^c)=0$, $E(\epsilon_i^r)=c>0$, and the signs ``$\pm$" correspond to the two cases in (\ref{mod-cases}). Define
\begin{equation}
  \begin{cases}
	\lambda_i=\epsilon_i^c,\ \eta_i=\epsilon_i^r,\ &\text{if}\ Y_i\ominus \left(aX_i+b\right)\ \text{exists};\\ 
	\lambda_i=-\epsilon_i^c,\ \eta_i=-\epsilon_i^r,\ &\text{if otherwise}\ \left(aX_i+b\right)\ominus Y_i\ \text{exists}.
	\end{cases}
\end{equation}
Our model is specified as
\begin{eqnarray}
  Y_i^c &=& aX_i^c+b+\lambda_i,\label{mod-1**}\\
  Y_i^r &=& \left|a\right|X_i^r+\eta_i,\label{mod-2**}
\end{eqnarray}
where $\text{E}(\lambda_i)=0$, $\text{E}(\eta_i)=\mu\in[-c,c]$, $\text{Var}(\lambda_i)=\sigma^2_{\lambda}>0$, and $\text{Var}(\eta_i)=\sigma^2_{\eta}>0$. %Although, with this specification, the error interval $\epsilon_i$ is no longer identifiable, relation (\ref{mod-1}) still holds in that 
%\begin{equation*}
%  \delta\left(Y_i, aX_i+b\right)=\sqrt{\lambda_i^2+\eta_i^2}=\sqrt{\lambda_i^2+\delta_i^2}
%	=\left\|\epsilon_i\right\|_2.
%\end{equation*}
%Therefore, model (\ref{mod-1**})-(\ref{mod-2**}) is viewed as a reduced form of (\ref{mod-1})-(\ref{mod-2}). 

To model the outcome intervals $Y_i=\left[\underline{Y_i}, \overline{Y_i}\right]$ by $p$ interval-valued predictors $X_{j,i}=\left[\underline{X_{j,i}}, \overline{X_{j,i}}\right]$, $i=1,\cdots,n$; $j=1,\cdots,p$, we consider the multivariate extension of (\ref{mod-1}):
\begin{eqnarray}
  &&\delta\left(Y_i, b+\sum_{j=1}^{p}a_jX_{j,i}\right)=\left\|\epsilon_i\right\|_2,\label{mmod-1}
\end{eqnarray}
which leads to the following center-radius specification 
\begin{eqnarray}
  Y_i^c &=& b+\sum_{j=1}^{p}a_jX_{j,i}^c+\lambda_i,\label{mmod-1**}\\
  Y_i^r &=& \sum_{j=1}^{p}\left|a_j\right|X_{j,i}^r+\eta_i.\label{mmod-2**}
\end{eqnarray}
where $\text{E}(\lambda_i)=0$, $\text{E}(\eta_i)=\mu\in[-c,c]$, $\text{Var}(\lambda_i)=\sigma^2_{\lambda}$, and $\text{Var}(\eta_i)=\sigma^2_{\eta}$. We have assumed $\lambda_i$ and $\eta_i$ are independent in this paper to simplify the presentation. The model that includes a covariance between $\lambda_i$ and $\eta_i$ can be implemented without much extra difficulty. 

%=======================================================================================================================%
\subsection{Least squares estimate (LSE)}
Least squares method is widely used in the literature to estimate the interval-valued regression coefficients (\cite{Diamond90}, \cite{Korner98}, \cite{Gil02}). It minimizes $\delta\left(Y,\text{E}(Y|X)\right)$ on the data with respect to the parameters. Denote 
\begin{eqnarray}
  \hat{Y}_i^c&=&\text{E}(Y_i^c|X_i)=b+\sum_{j=1}^{p}a_jX_{j,i}^c,\label{exp-c}\\
	\hat{Y}_i^r&=&\text{E}(Y_i^r|X_i)=\mu+\sum_{j=1}^{p}\left|a_j\right|X_{j,i}^r.\label{exp-r}
\end{eqnarray}
Then the sum of squared $\delta$-distance between $Y_i$ and $\hat{Y}_i$ is written as
\begin{eqnarray*}
  L&=&\sum_{i=1}^{n}\delta^2\left[\text{E}\left(Y_i|X_i\right), Y_i\right]\\
	&=&\sum_{i=1}^{n}\left[\left(b+\sum_{j=1}^{p}a_jX_{j,i}^c-Y_i^c\right)^2+\left(\sum_{j=1}^{p}\left|a_j\right|X_{j,i}^r+\mu-Y_i^r\right)^2\right].
\end{eqnarray*}
Therefore, the LSE of $\left\{\mu, b, a_j, j=1,\cdots, p\right\}$ is defined as
\begin{equation}
  \left\{\hat{\mu}, \hat{b}, \hat{a}_j, j=1,\cdots, p\right\}
  =\arg\min \left\{\frac{1}{n}L\left(\mu, b, a_j, j=0,\cdots, p\right)\right\}.\label{def-ls}
\end{equation}
Let 
\begin{eqnarray*}
  S\left(X_j^c, X_k^c\right) &=& \frac{1}{n}\sum_{i=1}^{n}X_{j,i}^cX_{k,i}^c-\left(\frac{1}{n}\sum_{i=1}^{n}X_{j,i}^c\right)\left(\frac{1}{n}\sum_{i=1}^{n}X_{k,i}^c\right),\\
  S\left(X_j^r, X_k^r\right) &=&  \frac{1}{n}\sum_{i=1}^{n}X_{j,i}^rX_{k,i}^r-\left(\frac{1}{n}\sum_{i=1}^{n}X_{j,i}^r\right)\left(\frac{1}{n}\sum_{i=1}^{n}X_{k,i}^r\right),
\end{eqnarray*}
be the sample covariances of the centers and radii of $X_j$ and $X_k$, respectively. Especially, when $k=j$, we denote by $S^2\left(X_j^c\right)$ and $S^2\left(X_j^r\right)$ the corresponding sample variances. In addition, define
\begin{eqnarray*}
  S\left(X_j^c, Y^c\right) &=& \frac{1}{n}\sum_{i=1}^{n}X_{j,i}^cY^c-\left(\frac{1}{n}\sum_{i=1}^{n}X_{j,i}^c\right)\left(\frac{1}{n}\sum_{i=1}^{n}Y^c\right),\\
  S\left(X_j^r, Y^r\right) &=& \frac{1}{n}\sum_{i=1}^{n}X_{j,i}^rY^r-\left(\frac{1}{n}\sum_{i=1}^{n}X_{j,i}^r\right)\left(\frac{1}{n}\sum_{i=1}^{n}Y^r\right),
\end{eqnarray*}
as the sample covariances of the centers and radii of $X_j$ and $Y$, respectively. Then, the minimization problem (\ref{def-ls}) is solved in the following proposition. 

%=====================================================Proposition ========================================================%
\begin{proposition}\label{prop:ls_solu}
The least squares estimates of the regression coefficients $\left\{\hat{a}_j\right\}_{j=1}^{p}$, if they exist, are solution of the equation system:
\begin{eqnarray}
  &&\sum_{j=1}^{p}a_jS\left(X_j^c, X_k^c\right)+sgn\left(a_k\right)\sum_{j=1}^{p}|a_j|S\left(X_j^r, X_k^r\right)\nonumber\\
  &=&S\left(X_k^c, Y^c\right)+sgn\left(a_k\right)S\left(X_k^r, Y^r\right),\ \ k=1,\cdots, p.\label{eqn:lse-1}
\end{eqnarray}
And then, $\left\{\hat{b}, \hat{\mu}\right\}$ are given by
\begin{eqnarray}
  \hat{b}&=&\overline{Y^c}-\sum_{j=1}^{p}\hat{a}_j\overline{X_j^c},\label{eqn:lse-2}\\
  \hat{\mu}&=&\overline{Y^r}-\sum_{j=1}^{p}|\hat{a}_j|\overline{X_j^r}.\label{eqn:lse-3}
\end{eqnarray}
\end{proposition}

%=======================================================================================================================%
\subsection{Sums of squares and $R^2$}
The variance of a compact convex random set $X$ in $\mathbb{R}^d$ is defined via its support function as 
\begin{equation*}
  \text{Var}(X)=\text{E}\delta^2\left(X, \text{E}X\right),
\end{equation*}
where the expectation is defined by Aumann integral (see \cite{Aumann65}, \cite{Artstein75}) as
\begin{eqnarray*}
  \text{E}X=\left\{\text{E}\xi:\xi\in X \text{ almost surely}\right\}.
\end{eqnarray*}
See \cite{Korner95,Korner97}. For the case $d=1$, it is shown by straightforward calculations that
\begin{eqnarray*}
  &&\text{E}X=[\text{E}\underline{X},\text{E}\overline{X}],\\
  &&\text{Var}(X)=\text{Var}\left(X^c\right)+\text{Var}\left(X^r\right).
\end{eqnarray*}
This leads us to define the sums of squares in $\mathcal{K}_\mathcal{C}\left(\mathbb{R}\right)$ to measure the variability of interval-valued data. A definition of the coefficient of determination $R^2$ in $\mathcal{K}_\mathcal{C}\left(\mathbb{R}\right)$ follows immediately, which produces a measure of goodness-of-fit. 
\begin{definition}
The total sum of squares (SST) in $\mathcal{K}_\mathcal{C}$ is defined as 
\begin{equation}\label{def:sst}
  SST=\sum_{i=1}^{n}\left[\left(Y_i^c-\overline{Y^c}\right)^2+\left(Y_i^r-\overline{Y^r}\right)^2\right]. 
\end{equation}
\end{definition}
\begin{definition}
The explained sum of squares (SSE) in $\mathcal{K}_\mathcal{C}$ is defined as 
\begin{equation}\label{def:sse}
  SSE=\sum_{i=1}^{n}\left[\left(\hat{Y}_i^c-\overline{Y^c}\right)^2+\left(\hat{Y}_i^r-\overline{Y^r}\right)^2\right]. 
\end{equation}
\end{definition}
\begin{definition}
The residual sum of squares (SSR) in $\mathcal{K}_\mathcal{C}$ is defined as 
\begin{equation}\label{def:ssr}
  SSR=\sum_{i=1}^{n}\left[\left(Y_i^c-\hat{Y}_i^c\right)^2+\left(Y_i^r-\hat{Y}_i^c\right)^2\right]. 
\end{equation}
\end{definition}
\begin{definition}
The coefficient of determination ($R^2$) in $\mathcal{K}_\mathcal{C}$ is defined as
\begin{equation}\label{def:r2}
  R^2=1-\frac{SSR}{SST},
\end{equation}
where $SST$ and $SSR$ are defined in (\ref{def:sst}) and (\ref{def:ssr}), respectively.  
\end{definition}
Analogous to the classical theory of linear regression, our model (\ref{mmod-1**})-(\ref{mmod-2**}) together with the LS estimates (\ref{def-ls}) accommodates the partition of $SST$ into $SSE$ and $SSR$. As a result, the coefficient of determination ($R^2$) can also be calculated as the ratio of $SSE$ and $SST$. The partition has a series of important implications of the underlying model, one of which being that the residual $Y\ominus\hat{Y}$/$\hat{Y}\ominus Y$ and the predictor $\hat{Y}$ are empirically uncorrelated in $\left(\mathcal{K}_{\mathcal{C}},\delta\right)$.  
%==================================================Theorem 1=======================================================%
\begin{theorem}\label{thm:ss}
Assume model (\ref{mmod-1**})-(\ref{mmod-2**}). Let $Y_i^c$ and $Y_i^r$ in (\ref{exp-c})-(\ref{exp-r}) be calculated according to the LS estimates $\left\{\hat{\mu}, \hat{b}, \hat{a}_j, j=1,\cdots, p\right\}$ in (\ref{def-ls}). Then, 
\begin{equation*}
  SST=SSE+SSR.
\end{equation*}
It follows that the coefficient of determination in $\mathcal{K}_{\mathcal{C}}$ is equivalent to
\begin{equation*}
  R^2=SSE/SST.
\end{equation*}
\end{theorem}

%====================================================================================================================%
\subsection{Positivity of $\hat{Y}_i^r$ and goodness-of-fit}
It is possible to get negative values of $\hat{Y}_i^r$ by its definition (\ref{exp-r}). Theorem \ref{thm:pred-adjust} gives an upper bound of the probability of this unfortunate event. If the model largely explains the variability of $Y^r$, $\sigma^2_{\eta}$ should be very small and so is this bound. Then, the rare cases of negative $\hat{Y}_i^r$ can be rounded up to 0 since $Y_i^r$ is nonnegative. Otherwise, if most of the variability of $Y^r$ lies in the random error, the probability of getting negative predicts may not be ignorable, but it is essentially due to the insufficiency of the model and a different model should be pursued anyway. 

%======================================================Theorem 1=====================================================%
\begin{theorem}\label{thm:pred-adjust}
Consider model (\ref{mmod-1**})-(\ref{mmod-2**}). Let $\hat{Y}_i$ be defined in (\ref{exp-c})-(\ref{exp-r}). Then,
\begin{equation*}
  P\left(\hat{Y}_i^r<0\right)
  \leq\frac{E\left(Y_i^r-\hat{Y_i^r}\right)^2}{\left(Y_i^r\right)^2}
  =\frac{\sigma^2_{\eta}}{\left(Y_i^r\right)^2}.
\end{equation*}
\end{theorem}

%=================================================Properties of LSE=================================================%
\section{Properties of LSE for the univariate model}
%===================================================================================================================%
In this section, we study the theoretical properties of the LSE for the univariate model (\ref{mod-1**})-(\ref{mod-2**}). Applying Proposition \ref{prop:ls_solu} to the case $p=1$, we obtain the two sets of half-space solutions, corresponding to $a\geq 0$ and $a<0$, respectively, as follows:
\begin{eqnarray}
    a^+ &=& \frac{S(X^c,Y^c)+S(X^r,Y^r)}{S^2(X^c)+S^2(X^r)},\label{a+}\\
    b^+ &=& \overline{Y^c}-a^+\overline{X^c},\label{b+}\\
   \mu^+ &=& \overline{Y^r}-|a^+|\overline{X^r};\label{mu+}
\end{eqnarray}
and
\begin{eqnarray}
   a^- &=& \frac{S(X^c,Y^c)-S(X^r,Y^r)}{S^2(X^c)+S^2(X^r)},\label{a-}\\
   b^- &=& \overline{Y^c}-a^-\overline{X^c},\label{b-}\\
   \mu^- &=& \overline{Y^r}-|a^-|\overline{X^r}.\label{mu-}
\end{eqnarray}
The final formula for the LS estimates falls in three categories. In the first, there is one and only one set of existing solution, which is defined as the LSE. In the second, both sets of solutions exist, and the LSE is the one that minimizes $L$. In the third situation, neither solution exists, but this only happens with probability going to $0$. We conclude these findings in the following Theorem. 

\begin{theorem}\label{thm:ls_solu}
Assume model (\ref{mod-1**})-(\ref{mod-2**}). Let $\left\{\hat{a}, \hat{b}, \hat{\mu}\right\}$ be the least squares solution defined in (\ref{def-ls}).  If $|S(X^c,Y^c)|>|S(X^r,Y^r)|$, then there exists one and only one half-space solution. More specifically,\\

\textbf{i.} if in addition $S\left(X^c, Y^c\right)>0$, then the LS solution is given by
\begin{equation*}
  \left\{\hat{a}, \hat{b}, \hat{\mu}\right\}=\left\{a^+, b^+, \mu^+\right\};
\end{equation*}

\textbf{ii.} if instead $S\left(X^c, Y^c\right)<0$, then the LS solution is given by
\begin{equation*}
  \left\{\hat{a}, \hat{b}, \hat{\mu}\right\}=\left\{a^-, b^-, \mu^-\right\}.
\end{equation*}
Otherwise, $|S(X^c,Y^c)|<|S(X^r,Y^r)|$, and then either both of the half-space solutions exist, or neither one exists. In particular,\\

\textbf{iii.} if in addition $S\left(X^r, Y^r\right)>0$, then both of the half-space solutions exist, and
\begin{equation*}
  \left\{\hat{a}, \hat{b}, \hat{\mu}\right\}={\arg\min}_ {\left\{\left\{a^+,b^+,\mu^+\right\},\left\{a^-,b^-,\mu^-\right\}\right\}}
  \left\{L\left(a,b,\mu\right)\right\};
\end{equation*}

\textbf{iv.} if instead $S\left(X^r, Y^r\right)<0$, then the LS solution does not exist, but this happens with probability converging to 0.  
\end{theorem}

%=======================================================================================================================%
%The LS method is widely used in the literature to estimate parameters for various interval linear regression models (e.g., Diamond 1990, Gil et al. 2002, K\"orner and N\"ather 1998, Blanco-Fern\'andez et al. 2011, 2012). However, a thorough systematic investigation of the theoretical properties is still non-existent. Blanco-Fern\'andez et al. (2011) showed that the LS estimators for their model coefficients are strongly consistent, which essentially validates large sample performances of the estimators. In this section, we examine our LS estimators defined in Theorem \ref{thm:ls_solu} from another perspective. 

Unlike the classical linear regression, LS estimates for the model (\ref{mod-1**})-(\ref{mod-2**}) are biased. We calculate the biases explicitly in Proposition \ref{prop:ls_exp}, which are shown to converge to zero as the sample size increases to infinity. Therefore, the LS estimates are asymptotically unbiased.
%Basically our results caution against the usage of the LS estimators for small sample sizes, and whether they are suited for moderate sample sizes depends on the vanishing rate of the biases. For practical needs, an empirical investigation of the moderate sample size performances will be presented in Section \ref{sec:simu}. 
\begin{proposition}\label{prop:ls_exp}
Let $\left\{\hat{a},\hat{b},\hat{\mu}\right\}$ be the least squares solution in Theorem \ref{thm:ls_solu}. Then,
\begin{eqnarray*}
  E\left(\hat{a}-a\right)=-\frac{2aS^2(X^r)}{S^2(X^c)+S^2(X^r)}
  \left[P(\hat{a}=a^-)I_{\left\{a\geq 0\right\}}+P(\hat{a}=a^+)I_{\left\{a<0\right\}}\right],
\end{eqnarray*}
\begin{eqnarray*}
  E\left(|\hat{a}|-|a|\right)=-\frac{2|a|S^2(X^c)}{S^2(X^c)+S^2(X^r)}
  \left[P(\hat{a}=a^-)I_{\left\{a\geq 0\right\}}+P(\hat{a}=a^+)I_{\left\{a<0\right\}}\right].
\end{eqnarray*}
\end{proposition}

\begin{theorem}\label{thm:ls_consist}
Consider model (\ref{mod-1**})-(\ref{mod-2**}). Assume $S^2\left(X^c\right)=O(1)$ and $S^2\left(X^r\right)=O(1)$. Then, the least squares solution $\left\{\hat{a},\hat{b},\hat{\mu}\right\}$ in Theorem \ref{thm:ls_solu} is asymptotically unbiased, i.e.
\begin{eqnarray*}
  E\begin{bmatrix} \hat{a}\\ \hat{b}\\ \hat{\mu}\end{bmatrix}\to
  \begin{bmatrix} a\\ b\\ \mu \end{bmatrix},
\end{eqnarray*}
as $n\to\infty$.
\end{theorem}

%=====================================================Simulation====================================================%
\section{Simulation}\label{sec:simu}
We carry out a systematic simulation study to examine the empirical performance of the least squares method proposed in this paper. First, we consider the following three models:\\

\begin{itemize}
\item Model 1: $a=2$, $b=5$, $\mu=0.5$, $\sigma_{\eta}=0.3$, $\sigma_{\lambda}=2$;
\item Model 2: $a=-2$, $b=5$, $\mu=0.5$, $\sigma_{\eta}=0.3$, $\sigma_{\lambda}=3$;
\item Model 3: $a=2$, $b=5$, $\mu=-0.5$, $\sigma_{\eta}=0.3$, $\sigma_{\lambda}=2$;
%\item Model 4: $a=2$, $b=5$, $\mu=3$, $\sigma_{\eta}=2$, $\sigma_{\lambda}=4$,
\end{itemize}

\noindent where data show a positive correlation, a negative correlation, and a positive correlation with a negative $\mu$, respectively. A simulated dataset from each model is shown in Figure \ref{fig:sim-data}, along with its fitted regression line.

%======================================================Figure 1==========================================================%
\begin{figure}[ht]
\centering
\includegraphics[ height=1.900in, width=2.300in]{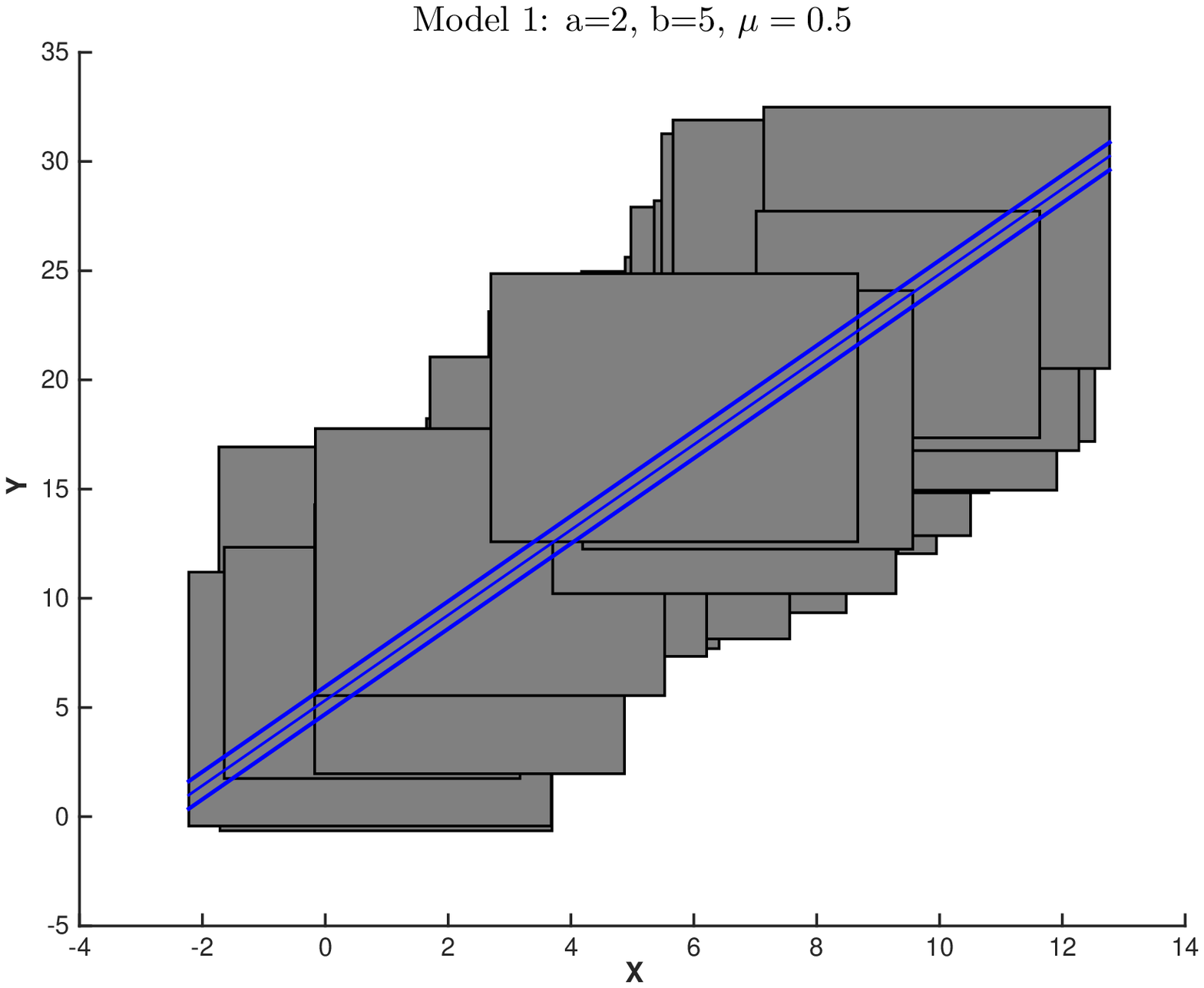}
\includegraphics[ height=1.900in, width=2.300in]{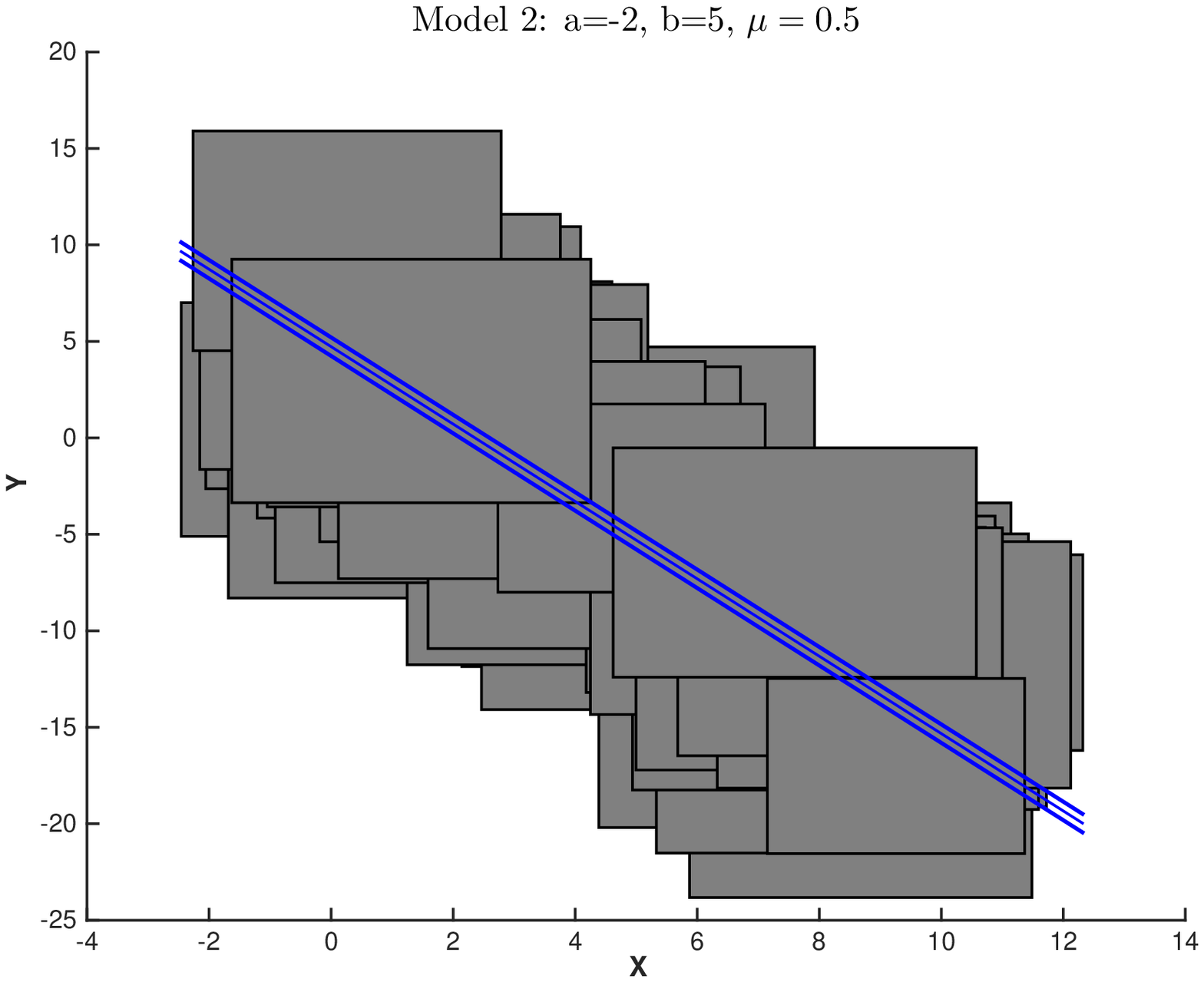}\\
\includegraphics[ height=1.900in, width=2.300in]{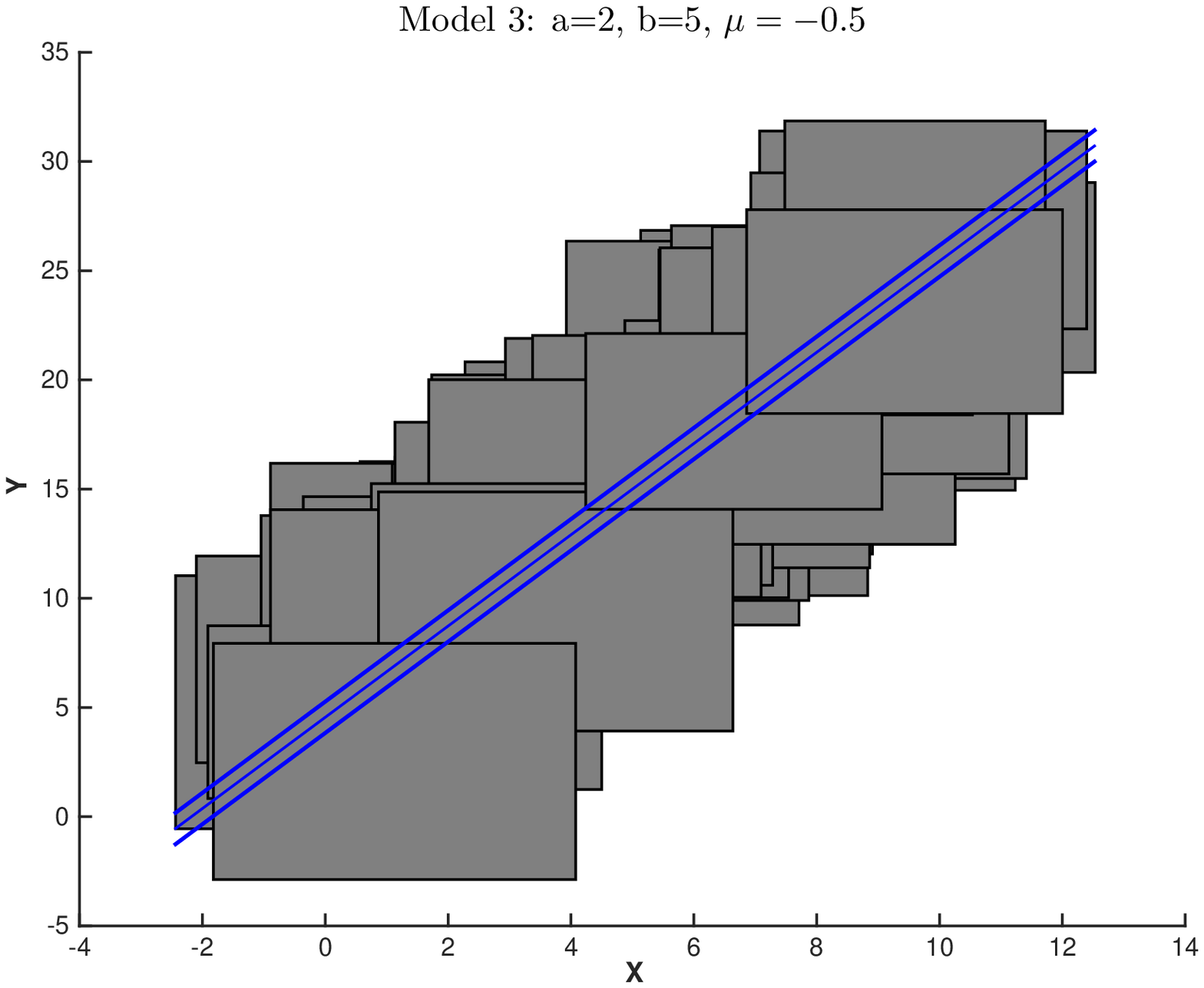}
\caption{Plots of simulated datasets from models 1, 2, and 3, each with sample size $n=50$. The solid line denotes the regression line $y=\hat{a}x+\hat{b}$, and the two dashed lines denote the two accompanying lines $y=\hat{a}x+\hat{b}\pm\hat{\mu}$.}
\label{fig:sim-data}
\end{figure}

To investigate the asymptotic behavior of the LS estimates, we repeat the process of data generation and parameter estimation 1000 times independently using sample size $n=20, 50, 100$ for all the three models. The resulting 1000 independent sets of parameter estimates for each model/sample size are evaluated by their mean absolute error (MAE) and mean error (ME). The numerical results are summarized in Table \ref{tab:sim}. Consistent with Proposition \ref{prop:ls_exp}, $\hat{a}$ tends to underestimate $a$ when $a>0$ and overestimate $a$ when $a<0$. This bias also causes a positive and negative bias in $\hat{b}$, when $a>0$ and $a<0$, respectively. Similarly, a positive bias in $\hat{\mu}$ is induced by the negative bias in $|\hat{a}|$. All the biases diminish to 0 as the sample size increases to infinity, which confirms our finding in Theorems \ref{thm:ls_consist}.

%======================================================Table 1==========================================================%
\begin{table}[htbp]
  \centering
  \caption{Evaluation of Parameter Estimation}
  \bigskip
    \begin{tabular}{crrrrrrr}
    \toprule
          & \multicolumn{1}{c}{n} & \multicolumn{1}{c}{MAE} & \multicolumn{1}{c}{ME} & \multicolumn{1}{c}{MAE} & \multicolumn{1}{c}{ME} & \multicolumn{1}{c}{MAE} & \multicolumn{1}{c}{ME} \\
    \midrule
          & \multicolumn{1}{c}{} & \multicolumn{1}{c}{} & \multicolumn{1}{c}{} & \multicolumn{1}{c}{} & \multicolumn{1}{c}{} & \multicolumn{1}{c}{} & \multicolumn{1}{c}{} \\
          & \multicolumn{1}{c}{} & \multicolumn{2}{c}{\textbf{a=2}} & \multicolumn{2}{c}{\textbf{b=5}} & \multicolumn{2}{c}{$\bdsm{\mu=0.5}$} \\
    Model 1 & \multicolumn{1}{l}{20} & \multicolumn{1}{l}{0.1449} & \multicolumn{1}{l}{-0.0921} & \multicolumn{1}{l}{0.8083} & \multicolumn{1}{l}{0.445} & \multicolumn{1}{l}{0.3655} & \multicolumn{1}{l}{0.2304} \\
          & \multicolumn{1}{l}{50} & \multicolumn{1}{l}{0.0848} & \multicolumn{1}{l}{-0.0411} & \multicolumn{1}{l}{0.4899} & \multicolumn{1}{l}{0.2141} & \multicolumn{1}{l}{0.214} & \multicolumn{1}{l}{0.1011} \\
          & \multicolumn{1}{l}{100} & \multicolumn{1}{l}{0.0562} & \multicolumn{1}{l}{-0.0171} & \multicolumn{1}{l}{0.3151} & \multicolumn{1}{l}{0.0872} & \multicolumn{1}{l}{0.142} & \multicolumn{1}{l}{0.041} \\
          & \multicolumn{1}{l}{} &       &       &       & \multicolumn{1}{l}{} & \multicolumn{1}{l}{} & \multicolumn{1}{l}{\textbf{}} \\
          & \multicolumn{1}{l}{} & \multicolumn{2}{c}{\textbf{a=-2}} & \multicolumn{2}{c}{\textbf{b=5}} & \multicolumn{2}{c}{$\bdsm{\mu=0.5}$} \\
    Model 2 & \multicolumn{1}{l}{20} & \multicolumn{1}{l}{0.2011} & \multicolumn{1}{l}{0.103} & \multicolumn{1}{l}{1.1389} & \multicolumn{1}{l}{-0.5071} & \multicolumn{1}{l}{0.5067} & \multicolumn{1}{l}{0.2578} \\
          & \multicolumn{1}{l}{50} & \multicolumn{1}{l}{0.1205} & \multicolumn{1}{l}{0.0336} & \multicolumn{1}{l}{0.6973} & \multicolumn{1}{l}{-0.1774} & \multicolumn{1}{l}{0.3038} & \multicolumn{1}{l}{0.0807} \\
          & \multicolumn{1}{l}{100} & \multicolumn{1}{l}{0.0842} & \multicolumn{1}{l}{0.0185} & \multicolumn{1}{l}{0.4814} & \multicolumn{1}{l}{-0.0865} & \multicolumn{1}{l}{0.2118} & \multicolumn{1}{l}{0.0465} \\
          & \multicolumn{1}{l}{} & \multicolumn{1}{l}{} & \multicolumn{1}{l}{} & \multicolumn{1}{l}{} & \multicolumn{1}{l}{} & \multicolumn{1}{l}{} & \multicolumn{1}{l}{\textbf{}} \\
          & \multicolumn{1}{l}{} & \multicolumn{2}{c}{\textbf{a=2}} & \multicolumn{2}{c}{\textbf{b=5}} & \multicolumn{2}{c}{$\bdsm{\mu=-0.5}$} \\
    Model 3 & \multicolumn{1}{l}{20} & \multicolumn{1}{l}{0.1488} & \multicolumn{1}{l}{-0.1047} & \multicolumn{1}{l}{0.8143} & \multicolumn{1}{l}{0.495} & \multicolumn{1}{l}{0.3785} & \multicolumn{1}{l}{0.262} \\
          & \multicolumn{1}{l}{50} & \multicolumn{1}{l}{0.0836} & \multicolumn{1}{l}{-0.0412} & \multicolumn{1}{l}{0.4703} & \multicolumn{1}{l}{0.2119} & \multicolumn{1}{l}{0.2108} & \multicolumn{1}{l}{0.1015} \\
          & \multicolumn{1}{l}{100} & \multicolumn{1}{l}{0.0579} & \multicolumn{1}{l}{-0.0187} & \multicolumn{1}{l}{0.3321} & \multicolumn{1}{l}{0.098} & \multicolumn{1}{l}{0.1453} & \multicolumn{1}{l}{0.0464} \\
    \bottomrule
    \end{tabular}
  \label{tab:sim}
\end{table}

Next, we compare our model to CCRM, a typical bivariate type of model from the literature. As we discussed in the introduction, these two models are developed for different purposes and are generally not comparable. We include a comparison in the simulation study to better evaluate the performances of our model, with CCRM providing a baseline of converging rate and predicting accuracy. From Model 1, 2, 3, respectively, we simulate 1000 independent samples with size $n=20, 50, 100$. Then, each sample is randomly split into a training set (80\%) and a validation set (20\%). The two models are evaluated by their sample variance adjusted mean squared errors (AMSE's) on the validation set, which are defined as
\begin{eqnarray*}
  &&\text{AMSE(center)}=\frac{\sum_{i=1}^{m}\left(Y^c_i-\hat{Y}^c_i\right)^2}
	{\sum_{i=1}^{m}\left(Y^c_i-\overline{Y}^c_i\right)^2},\\ 
  &&\text{AMSE(radius)}=\frac{\sum_{i=1}^{m}\left(Y^r_i-\hat{Y}^r_i\right)^2}
	{\sum_{i=1}^{m}\left(Y^r_i-\overline{Y}^r_i\right)^2},
\end{eqnarray*}
and 
\begin{eqnarray*}
  \text{AMSE(average)}=\frac{\text{AMSE(center)}+\text{AMSE(radius)}}{2},
\end{eqnarray*}
where $m=n/5$ is the size of validation set. We use the R function $ccrm$ in the iRegression package to implement CCRM. The average result of the 1000 repetitions are summarized in Table \ref{tab:sim-com}. For Model 1 and 2, both models have competitive performances. Model 3 has a negative $\mu$, so CCRM is slightly worse than our model due to its positive restriction on $\mu$. To better show this, we continue to consider the following two univariate models and one multivariate model with a much smaller $\mu$:
\begin{itemize}
\item Model 4: $a=3$, $b=5$, $\mu=-5$, $\sigma_{\eta}=0.5$, $\sigma_{\lambda}=5$;
\item Model 5: $a=-3$, $b=5$, $\mu=-5$, $\sigma_{\eta}=0.5$, $\sigma_{\lambda}=5$;
\item Model 6: $a_1=-3$, $a_2=2$ $b=5$, $\mu=-5$, $\sigma_{\eta}=0.5$, $\sigma_{\lambda}=5$.
\end{itemize}
A sample of $n=50$ from each of Model 4 and 5 are plotted in Figure \ref{fig:sim-data-45}. For all of the three models, our model performs significantly better than CCRM. 

%======================================================Figure 2==========================================================%
\begin{figure}[ht]
\centering
\includegraphics[ height=1.900in, width=2.300in]{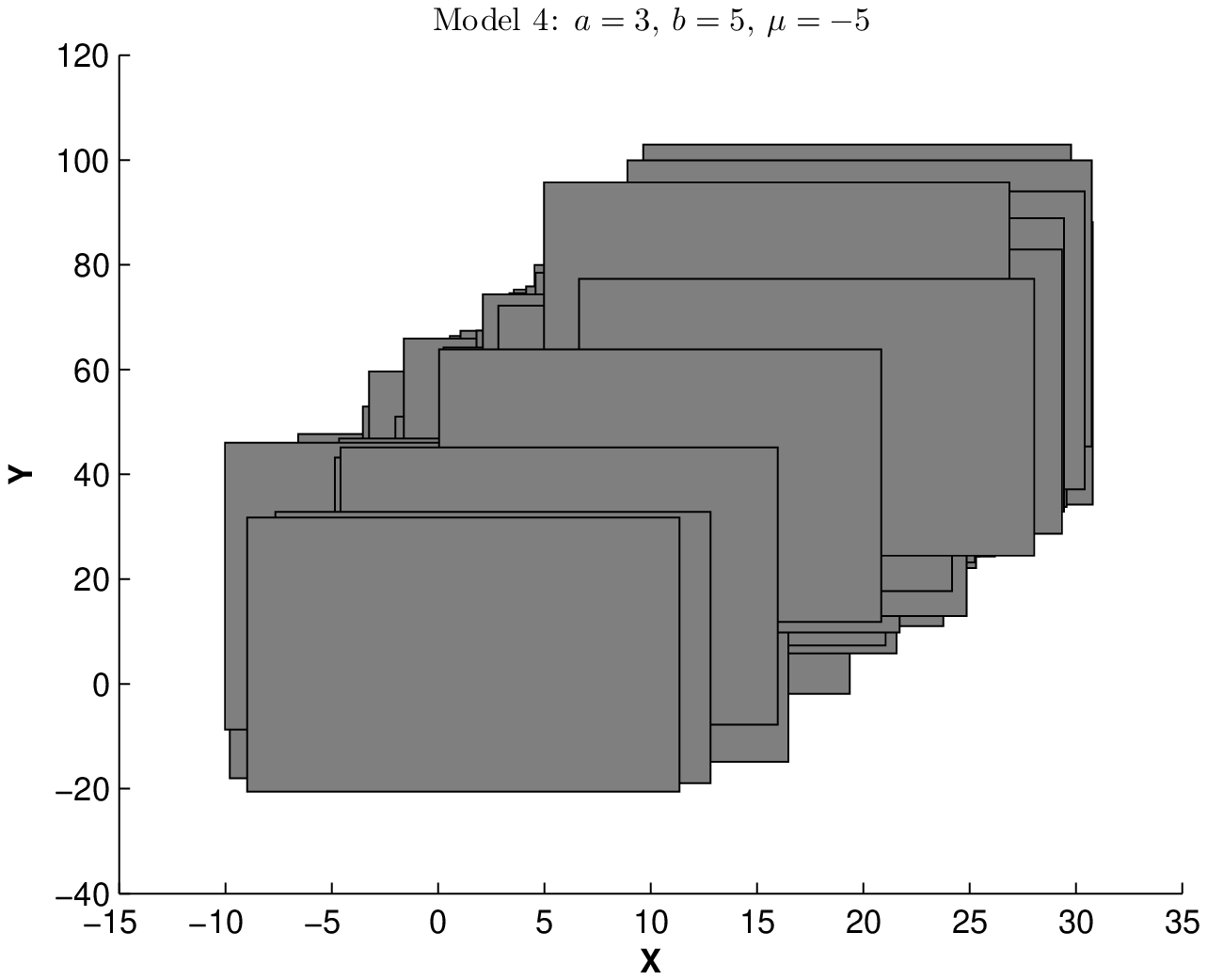}
\includegraphics[ height=1.900in, width=2.300in]{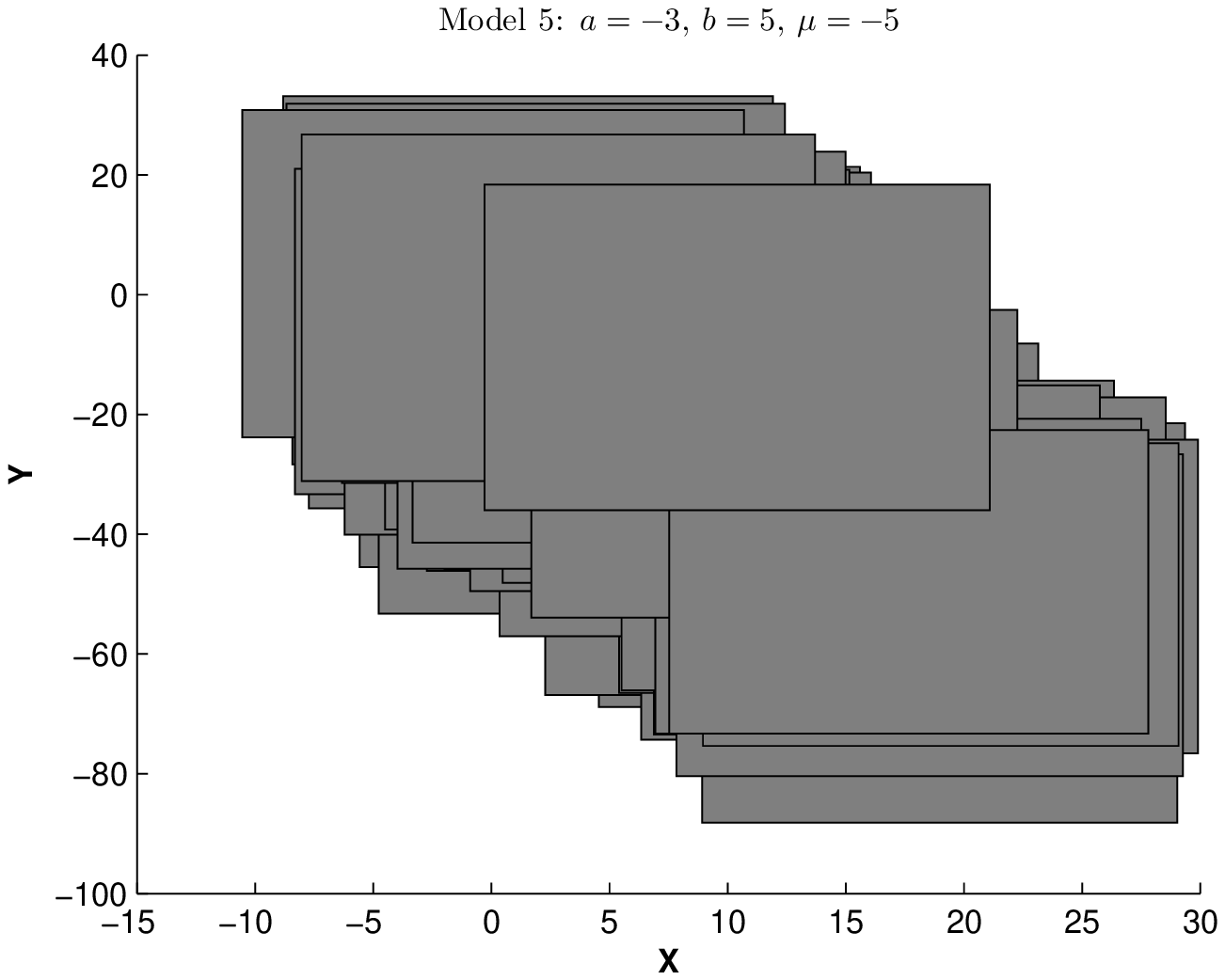}
\caption{Plots of simulated datasets from models 4 and 5, each with sample size $n=50$.}
\label{fig:sim-data-45}
\end{figure}

%======================================================Table 2==========================================================%
\begin{table}[htbp]
  \centering
  \caption{Mean results of AMSE on the validation set based on 1000 independent repetitions.}
  \bigskip
    \begin{tabular}{crrrrrrrr}
    \toprule
          & \multicolumn{1}{c}{} & \multicolumn{3}{c}{\textbf{CCRM}} & \multicolumn{1}{c}{\textbf{}} & \multicolumn{3}{c}{\textbf{Our Model}} \\
    \midrule
          & \multicolumn{1}{c}{n} & \multicolumn{1}{c}{Center} & \multicolumn{1}{c}{Radius} & \multicolumn{1}{c}{Average} & \multicolumn{1}{c}{} & \multicolumn{1}{c}{Center} & \multicolumn{1}{c}{Radius} & \multicolumn{1}{c}{Average} \\
          & \multicolumn{1}{c}{} & \multicolumn{1}{c}{} & \multicolumn{1}{c}{} & \multicolumn{1}{c}{} & \multicolumn{1}{c}{} & \multicolumn{1}{c}{} & \multicolumn{1}{c}{} & \multicolumn{1}{c}{} \\
    Model 1 & \multicolumn{1}{l}{20} & \multicolumn{1}{l}{0.1716} & \multicolumn{1}{l}{0.3134} & \multicolumn{1}{l}{\textbf{0.2425}} & \multicolumn{1}{l}{\textbf{}} & \multicolumn{1}{l}{0.1772} & \multicolumn{1}{l}{0.3374} & \multicolumn{1}{l}{0.2573} \\
          & \multicolumn{1}{l}{50} & \multicolumn{1}{l}{0.1181} & \multicolumn{1}{l}{0.2368} & \multicolumn{1}{l}{0.1775} & \multicolumn{1}{l}{} & \multicolumn{1}{l}{0.1116} & \multicolumn{1}{l}{0.2313} & \multicolumn{1}{l}{\textbf{0.1714}} \\
          & \multicolumn{1}{l}{100} & \multicolumn{1}{l}{0.1149} & \multicolumn{1}{l}{0.2241} & \multicolumn{1}{l}{0.1695} & \multicolumn{1}{l}{} & \multicolumn{1}{l}{0.1119} & \multicolumn{1}{l}{0.2219} & \multicolumn{1}{l}{\textbf{0.1669}} \\
          & \multicolumn{1}{l}{} & \multicolumn{1}{l}{} & \multicolumn{1}{l}{} & \multicolumn{1}{l}{} & \multicolumn{1}{l}{} & \multicolumn{1}{l}{} & \multicolumn{1}{l}{} & \multicolumn{1}{l}{} \\
    Model 2 & \multicolumn{1}{l}{20} & \multicolumn{1}{l}{0.3499} & \multicolumn{1}{l}{0.3244} & \multicolumn{1}{l}{\textbf{0.3372}} & \multicolumn{1}{l}{\textbf{}} & \multicolumn{1}{l}{0.3467} & \multicolumn{1}{l}{0.3294} & \multicolumn{1}{l}{0.338} \\
          & \multicolumn{1}{l}{50} & \multicolumn{1}{l}{0.2341} & \multicolumn{1}{l}{0.2356} & \multicolumn{1}{l}{0.2348} & \multicolumn{1}{l}{} & \multicolumn{1}{l}{0.2344} & \multicolumn{1}{l}{0.2318} & \multicolumn{1}{l}{\textbf{0.2331}} \\
          & \multicolumn{1}{l}{100} & \multicolumn{1}{l}{0.2263} & \multicolumn{1}{l}{0.2201} & \multicolumn{1}{l}{0.2232} & \multicolumn{1}{l}{} & \multicolumn{1}{l}{0.2203} & \multicolumn{1}{l}{0.22} & \multicolumn{1}{l}{\textbf{0.2201}} \\
          & \multicolumn{1}{l}{} & \multicolumn{1}{l}{} & \multicolumn{1}{l}{} & \multicolumn{1}{l}{} & \multicolumn{1}{l}{} & \multicolumn{1}{l}{} & \multicolumn{1}{l}{} & \multicolumn{1}{l}{} \\
    Model 3 & \multicolumn{1}{l}{20} & \multicolumn{1}{l}{0.1708} & \multicolumn{1}{l}{0.3367} & \multicolumn{1}{l}{0.2538} & \multicolumn{1}{l}{} & \multicolumn{1}{l}{0.1687} & \multicolumn{1}{l}{0.3241} & \multicolumn{1}{l}{\textbf{0.2464}} \\
          & \multicolumn{1}{l}{50} & \multicolumn{1}{l}{0.1192} & \multicolumn{1}{l}{0.2288} & \multicolumn{1}{l}{0.174} & \multicolumn{1}{l}{} & \multicolumn{1}{l}{0.1192} & \multicolumn{1}{l}{0.2246} & \multicolumn{1}{l}{\textbf{0.1719}} \\
          & \multicolumn{1}{l}{100} & \multicolumn{1}{l}{0.1128} & \multicolumn{1}{l}{0.2196} & \multicolumn{1}{l}{0.1662} & \multicolumn{1}{l}{} & \multicolumn{1}{l}{0.1101} & \multicolumn{1}{l}{0.219} & \multicolumn{1}{l}{\textbf{0.1646}} \\
          & \multicolumn{1}{l}{} & \multicolumn{1}{l}{} & \multicolumn{1}{l}{} & \multicolumn{1}{l}{} & \multicolumn{1}{l}{} & \multicolumn{1}{l}{} & \multicolumn{1}{l}{} & \multicolumn{1}{l}{\textbf{}} \\
    Model 4 & \multicolumn{1}{l}{20} & \multicolumn{1}{l}{0.3795} & \multicolumn{1}{l}{0.3499} & \multicolumn{1}{l}{0.3647} & \multicolumn{1}{l}{} & \multicolumn{1}{l}{0.125} & \multicolumn{1}{l}{0.3691} & \multicolumn{1}{l}{\textbf{0.247}} \\
          & \multicolumn{1}{l}{50} & \multicolumn{1}{l}{0.2802} & \multicolumn{1}{l}{0.2738} & \multicolumn{1}{l}{0.277} & \multicolumn{1}{l}{} & \multicolumn{1}{l}{0.0867} & \multicolumn{1}{l}{0.2734} & \multicolumn{1}{l}{\textbf{0.18}} \\
          & \multicolumn{1}{l}{100} & \multicolumn{1}{l}{0.258} & \multicolumn{1}{l}{0.2727} & \multicolumn{1}{l}{0.2653} & \multicolumn{1}{l}{} & \multicolumn{1}{l}{0.0808} & \multicolumn{1}{l}{0.2605} & \multicolumn{1}{l}{\textbf{0.1706}} \\
          & \multicolumn{1}{l}{} & \multicolumn{1}{l}{} & \multicolumn{1}{l}{} & \multicolumn{1}{l}{} & \multicolumn{1}{l}{} & \multicolumn{1}{l}{} & \multicolumn{1}{l}{} & \multicolumn{1}{l}{\textbf{}} \\
    Model 5 & \multicolumn{1}{l}{20} & \multicolumn{1}{l}{0.3519} & \multicolumn{1}{l}{0.3207} & \multicolumn{1}{l}{0.3363} & \multicolumn{1}{l}{} & \multicolumn{1}{l}{0.1204} & \multicolumn{1}{l}{0.3717} & \multicolumn{1}{l}{\textbf{0.2461}} \\
          & \multicolumn{1}{l}{50} & \multicolumn{1}{l}{0.2712} & \multicolumn{1}{l}{0.2799} & \multicolumn{1}{l}{0.2756} & \multicolumn{1}{l}{} & \multicolumn{1}{l}{0.0827} & \multicolumn{1}{l}{0.2681} & \multicolumn{1}{l}{\textbf{0.1754}} \\
          & \multicolumn{1}{l}{100} & \multicolumn{1}{l}{0.2558} & \multicolumn{1}{l}{0.2751} & \multicolumn{1}{l}{0.2655} & \multicolumn{1}{l}{} & \multicolumn{1}{l}{0.08} & \multicolumn{1}{l}{0.2552} & \multicolumn{1}{l}{\textbf{0.1676}} \\
          & \multicolumn{1}{l}{} & \multicolumn{1}{l}{} & \multicolumn{1}{l}{} & \multicolumn{1}{l}{} & \multicolumn{1}{l}{} & \multicolumn{1}{l}{} & \multicolumn{1}{l}{} & \multicolumn{1}{l}{\textbf{}} \\
    Model 6 & \multicolumn{1}{l}{50} & \multicolumn{1}{l}{0.0622} & \multicolumn{1}{l}{0.4288} & \multicolumn{1}{l}{0.2455} & \multicolumn{1}{l}{} & \multicolumn{1}{l}{0.0661} & \multicolumn{1}{l}{0.2536} & \multicolumn{1}{l}{\textbf{0.1599}} \\
          & \multicolumn{1}{l}{100} & \multicolumn{1}{l}{0.0596} & \multicolumn{1}{l}{0.3934} & \multicolumn{1}{l}{0.2265} & \multicolumn{1}{l}{} & \multicolumn{1}{l}{0.0606} & \multicolumn{1}{l}{0.237} & \multicolumn{1}{l}{\textbf{0.1488}} \\
          & \multicolumn{1}{l}{200} & \multicolumn{1}{l}{0.0565} & \multicolumn{1}{l}{0.3838} & \multicolumn{1}{l}{0.2201} & \multicolumn{1}{l}{} & \multicolumn{1}{l}{0.0593} & \multicolumn{1}{l}{0.2344} & \multicolumn{1}{l}{\textbf{0.1469}} \\
    \bottomrule
    \end{tabular}%
  \label{tab:sim-com}%
\end{table}%

%=====================================================Application========================================================%
\section{A real data application}
In this section, we apply our model to analyze the average temperature data for large US cities, which are provided by National Oceanic and Atmospheric Administration (NOAA) and are publicly available. The three data sets we obtained specifically are average temperatures for 51 large US cities in January, April, and July. Each observation contains the averages of minimum and maximum temperatures based on weather data collected from 1981 to 2010 by the NOAA National Climatic Data Center of the United States. July in general is the hottest month in the US. By this analysis, we aim to predict the summer (July) temperatures by those in the winter (January) and spring (April). Figure \ref{fig:real-data} plots the July temperatures versus those in January and April, respectively.   

%=======================================================Figure 3==========================================================%
\begin{figure}[ht]
\centering
\includegraphics[ height=2.000in, width=2.000in]{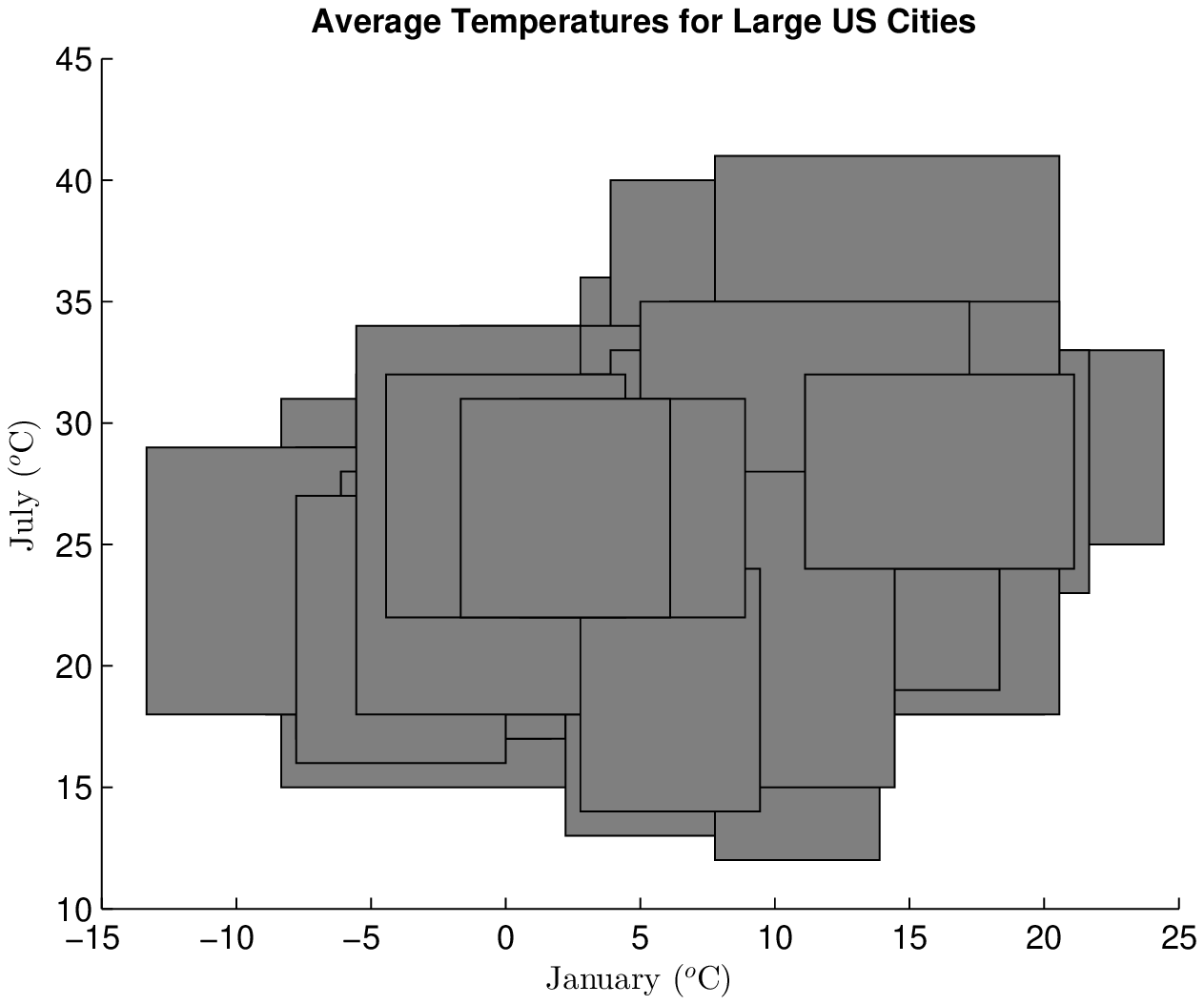}
\includegraphics[ height=2.000in, width=2.000in]{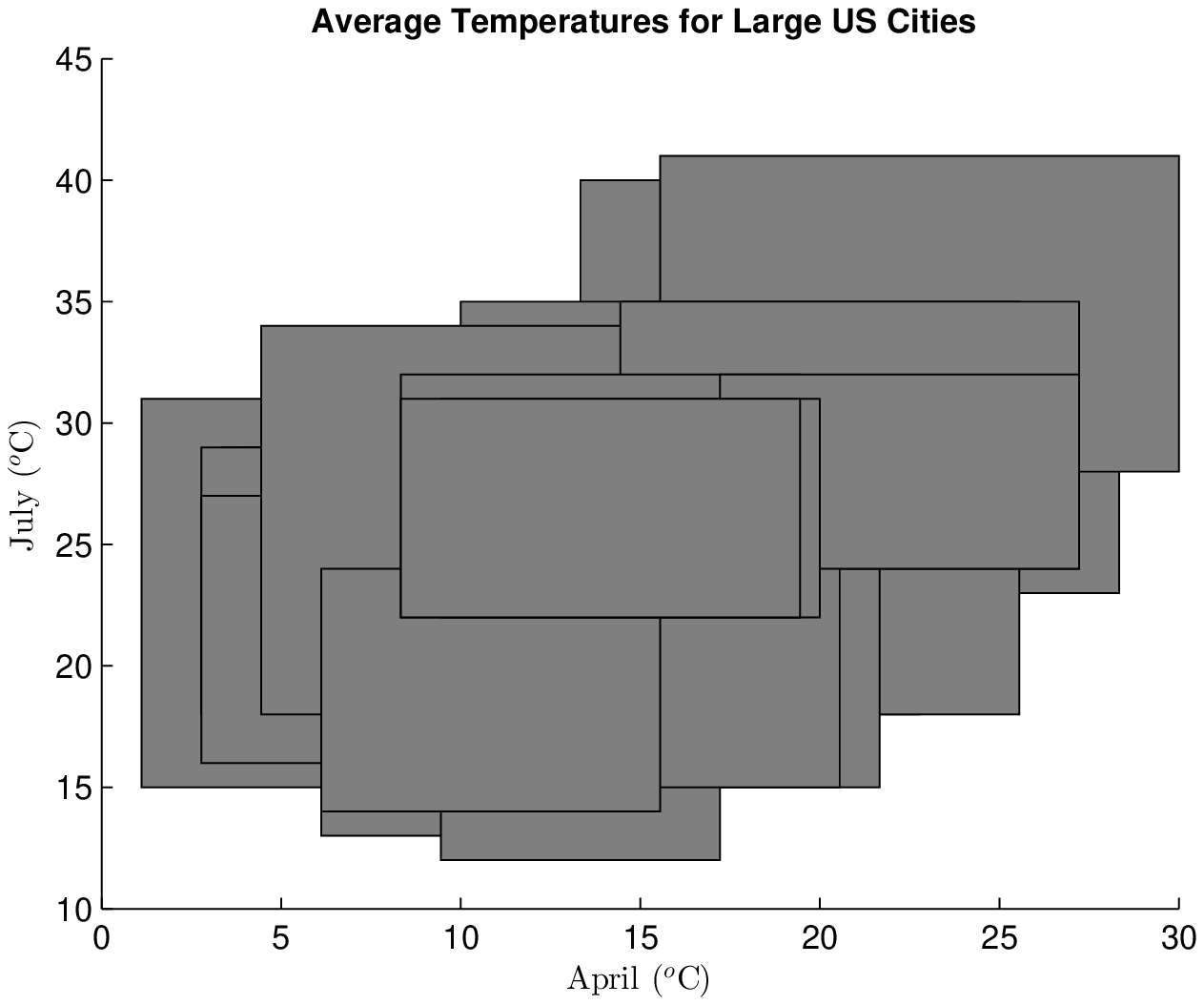}
\caption{Left: plot of July versus January temperatures. Right: plot of July versus April temperatures.}
\label{fig:real-data}
\end{figure}

The parameters are estimated according to (\ref{eqn:lse-1})-(\ref{eqn:lse-3}) as
\begin{eqnarray*}
  &&\hat{a}_1=-0.4831,\ \hat{a}_2=1.1926;\\
  &&\hat{b}=10.2510,\ \hat{\mu}=-3.7071.
\end{eqnarray*}
Denote by $T_{Jan}$, $T_{April}$, and $T_{July}$, the average temperatures in a US city in January, April, and July, respectively. %the fitted linear regression model is
%\begin{equation*}
%   \delta\left(T_{July}, 10.2510-0.4831 T_{Jan}+1.1926 T_{April}\right)=\left\|\epsilon\right\|_2,
%\end{equation*}
%where $\epsilon=\left[\lambda-\delta, \lambda+\delta\right]$ is the error interval whose distribution unfortunately is not completely recovered by the data. 
The prediction for $T_{July}$ based on $T_{Jan}$ and $T_{April}$ is given by
\begin{eqnarray}
  \hat{T}^c_{July}&=&10.2510-0.4831 T^c_{Jan}+1.1926 T^c_{April},\label{July_pred_c}\\
  \hat{T}^r_{July}&=&-3.7071+0.4831 T^r_{Jan}+1.1926 T^r_{April}.\label{July_pred_r}
\end{eqnarray}
The three sums of squares are calculated to be  
\begin{eqnarray*}
  SST=663.8627;\ SSE=495.0874;\ SSR=168.7753.
\end{eqnarray*}
Therefore, the coefficient of determination is 
\begin{eqnarray*}
  R^2=1-\frac{SSR}{SST}=\frac{SSE}{SST}=0.7458. 
\end{eqnarray*} 
Finally, the variance parameters can be estimated as
\begin{eqnarray*}
  \hat{\sigma}^2_{\lambda}&=&\frac{1}{n-1}\sum_{i=1}^{n}\left(T_{July,i}^c-\hat{T}_{July,i}^c\right)^2=2.1708;\\
  \hat{\sigma}^2_{\eta}&=&\frac{1}{n-1}\sum_{i=1}^{n}\left(T_{July,i}^r-\hat{T}_{July,i}^r\right)^2=1.2047.
\end{eqnarray*}
Thus, by Theorem \ref{thm:pred-adjust}, an upper bound of $P\left(\hat{T}^r_{July,i}<0\right)$ on average is estimated to be
\begin{equation*}
  \frac{1}{n}\sum_{i=1}^{n}\frac{\hat{\sigma}^2_{\eta}}{\left(T^r_{July,i}\right)^2}=\frac{1.2047}{n}\sum_{i=1}^{n}\frac{1}{\left(T^r_{July,i}\right)^2}=0.047,
\end{equation*}
which is very small and reasonably ignorable. We calculate $\hat{T}^r_{July,i}$ for the entire sample and all of them are well above $0$. So, for this data, although $\hat{\mu}<0$ and it is possible to get negative predicted radius, it in fact never happens because the model has captured most of the variability. The empirical distributions of residuals are shown in Figure \ref{fig:real-residual}. Both distributions are centered at 0, with the center residual having a slightly bigger tail.  

%====================================================Figure 4=======================================================%
\begin{figure}[ht]
\centering
\includegraphics[ height=2.200in, width=2.500in]{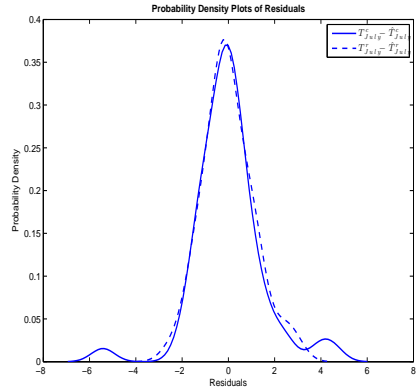}
\caption{Empirical probability density plots of the residuals for the center and radius. }
\label{fig:real-residual}
\end{figure}

%====================================================Conclusion======================================================%
\section{Concluding remarks}
We have rigorously studied linear regression for interval-valued data in the metric space $\left(\mathcal{K}_{\mathcal{C}}, \delta\right)$. The new model we introduces generalizes previous models in the literature so that the Hukuhara difference $Y_i\ominus \left(aX_i+b\right)$ needs not exist. Analogous to the classical linear regression, our model together with the LS estimation leads to a partition of the total sum of squares (SSR) into the explained sum of squares (SSE) and the residual sum of squares (SSR) in $\left(\mathcal{K}_{\mathcal{C}}, \delta\right)$, which implies that the residual is uncorrelated with the linear predictor in $\left(\mathcal{K}_{\mathcal{C}}, \delta\right)$. In addition, we have carried out theoretical investigations into the least squares estimation for the univariate model. It is shown that the LS estimates in $\left(\mathcal{K}_{\mathcal{C}}, \delta\right)$ are biased but the biases reduce to zero as the sample size tends to infinity. Therefore, a bias-correction technique for small sample estimation could be a good future topic. The simulation study confirms our theoretical findings and shows that the least squares estimators perform satisfactorily well for moderate sample sizes.

%====================================================Appendices=====================================================%
%\appendices
\section{Appendix: Proofs}
%===================================================================================================================%
\subsection{Proof of Proposition \ref{prop:ls_solu}}
\begin{proof}
Differentiating $L$ with respect to $\mu$, $b$, and $a_j$, $j=1,\cdots,p$, respectively, and setting the derivatives to zero, we get
\begin{eqnarray}
  &&\frac{\partial L}{\partial \mu}\propto\sum_{i=1}^{n}\left(\hat{Y}_i^r-Y_i^r\right)=0,\label{lse-1}\\
  &&\frac{\partial L}{\partial b}\propto\sum_{i=1}^{n}\left(\hat{Y}_i^c-Y_i^c\right)=0,\label{lse-2}\\
  &&\frac{\partial L}{\partial a_k}\propto\sum_{i=1}^{n}\left(\hat{Y}_i^c-Y_i^c\right)X_{k,i}^c
  +\sum_{i=1}^{n}\left(\hat{Y}_i^r-Y_i^r\right)sgn\left(a_k\right)X_{k,i}^r=0,\label{lse-3}\\
  && k=1,\cdots, p.\nonumber
\end{eqnarray}
Equations (\ref{lse-1})-(\ref{lse-2}) yield
\begin{eqnarray}
  b&=&\frac{1}{n}\sum_{i=1}^{n}Y_i^c-\frac{1}{n}\sum_{j=1}^{p}a_j\sum_{i=1}^{n}X_{j,i}^{c}
	=\overline{Y^c}-\sum_{j=1}^{p}a_j\overline{X_{j}^c},\label{lse-4}\\
	\mu&=&\frac{1}{n}\sum_{i=1}^{n}Y_i^r-\frac{1}{n}\sum_{j=1}^{p}|a_j|\sum_{i=1}^{n}X_{j,i}^{r}
	=\overline{Y^r}-\sum_{j=1}^{p}|a_j|\overline{X_{j}^r}.\label{lse-5}
\end{eqnarray}
Equations (\ref{eqn:lse-1}) are obtained by plugging (\ref{lse-4})-(\ref{lse-5}) into (\ref{lse-3}), and equations (\ref{eqn:lse-2})-(\ref{eqn:lse-3}) follow from (\ref{lse-4})-(\ref{lse-5}). This completes the proof. 
\end{proof}

%========================================================================================================================%
\subsection{Proof of Theorem \ref{thm:ss}}
\begin{proof}
According to definitions (\ref{def:sst})-(\ref{def:ssr}),
\begin{eqnarray}
  SST&=&\sum_{i=1}^{n}\left[\left(Y_i^c-\hat{Y}_i^c+\hat{Y}_i^c-\overline{Y^c}\right)^2+\left(Y_i^r-\hat{Y}_i^r+\hat{Y}_i^r-\overline{Y^r}\right)^2\right]\nonumber\\
  &=&SSE+SSR+2\sum_{i=1}^{n}\left[\left(Y_i^c-\hat{Y}_i^c\right)\left(\hat{Y}_i^c-\overline{Y^c}\right)
  +\left(Y_i^r-\hat{Y}_i^r\right)\left(\hat{Y}_i^r-\overline{Y^r}\right)\right]\nonumber\\
  &=&SSE+SSR+2\sum_{i=1}^{n}\left[\left(Y_i^c-\hat{Y}_i^c\right)\hat{Y}_i^c+\left(Y_i^r-\hat{Y}_i^r\right)\hat{Y}_i^r\right].\label{ss:eqn-1}
\end{eqnarray}
The last equation is due to (\ref{lse-1})-(\ref{lse-2}). Further in view of (\ref{exp-c})-(\ref{exp-r}) and (\ref{lse-3}), we have
\begin{eqnarray*}
  &&\sum_{i=1}^{n}\left[\left(Y_i^c-\hat{Y}_i^c\right)\hat{Y}_i^c+\left(Y_i^r-\hat{Y}_i^r\right)\hat{Y}_i^r\right]\\
  &=&\sum_{i=1}^{n}\left[\left(Y_i^c-\hat{Y}_i^c\right)\sum_{j=1}^{p}a_jX_{j,i}^c+\left(Y_i^r-\hat{Y}_i^r\right)\sum_{j=1}^{p}|a_j|X_{j,i}^r\right]\\
  &=&\sum_{j=1}^{p}a_j\sum_{i=1}^{n}\left[\left(Y_i^c-\hat{Y}_i^c\right)X_{j,i}^c+\left(Y_i^r-\hat{Y}_i^r\right)sgn(a_j)X_{j,i}^r\right]\\
  &=&0.
\end{eqnarray*}
This together with (\ref{ss:eqn-1}) completes the proof. 
\end{proof}

%========================================================================================================================%
\subsection{Proof of Theorem \ref{thm:pred-adjust}}
\begin{proof}
Notice that 
\begin{equation*}
  P\left(\hat{Y}_i^r<0\right)=P\left(\hat{Y}_i^r-Y_i^r<-Y_i^r\right)
  \leq P\left(|\hat{Y}_i^r-Y_i^r|>Y_i^r\right).
\end{equation*}
An application of Markov's inequality completes the proof.
\end{proof}

%========================================================================================================================%
\subsection{Proof of Theorem \ref{thm:ls_solu}}
\begin{proof}
Parts \textbf{i}, \textbf{ii} and \textbf{iii} are obvious from Proposition \ref{prop:ls_solu}. Part \textbf{iv} follows from Lemma \ref{lem:cov_r} in Appendix II.
\end{proof}

%========================================================================================================================%
\subsection{Proof of Proposition \ref{prop:ls_exp}}
\begin{proof}
We prove the cases $a\geq 0$ and $a<0$ separately. To simplify notations, we will use $E\left(\cdot\right)$ throughout the proof, but the expectation should be interpreted as being conditioned on $X$.\\

\textbf{Case I: $a\geq 0$}.\\

From Lemma \ref{lem:cov_est}, we have
\begin{align*}
&a^+-a\\
=&\frac{\sum_{i<j}(X_i^c-X_j^c)(Y_i^c-Y_j^c)+\sum_{i<j}(X_i^r-X_j^r)(Y_i^r-Y_j^r)}{\sum_{i<j}(X_i^c-X_j^c)^2+\sum_{i<j}(X_i^r-X_j^r)^2}-a\\
=&\frac{\sum_{i<j}(X_i^c-X_j^c)\left[ (Y_i^c-Y_j^c)-a(X_i^c-X_j^c) \right] +\sum_{i<j}(X_i^r-X_j^r) \left[ (Y_i^r-Y_j^r)-a(X_i^r-X_j^r) \right]}{\sum_{i<j}(X_i^c-X_j^c)^2+\sum_{i<j}(X_i^r-X_j^r)^2}\\
=&\frac{\sum_{i<j}(X_i^c-X_j^c)(\lambda_i-\lambda_j)+\sum_{i<j}(X_i^r-X_j^r)(\eta_i-\eta_j)}
{\sum_{i<j}(X_i^c-X_j^c)^2+\sum_{i<j}(X_i^r-X_j^r)^2}.\\
\end{align*}
This immediately yields 
\begin{equation}\label{eqn-1}
  E\left(a^+-a\right)=0. 
\end{equation}  
Similarly, 
\[
a^--a=\frac{\sum_{i<j}(X_i^c-X_j^c)(\lambda_i-\lambda_j)-\sum_{i<j}(X_i^r-X_j^r)[2a(X_i^r-X_j^r)+(\eta_i-\eta_j)]}
{\sum_{i<j}(X_i^c-X_j^c)^2+\sum_{i<j}(X_i^r-X_j^r)^2},
\]
and consequently,
\begin{equation}\label{eqn-2}
  E\left(a^--a\right)=-\frac{2aS^2(X^r)}{S^2(X^c)+S^2(X^r)}.
\end{equation}
Notice now 
\begin{eqnarray}
E(\hat{a}-a)
&=&E(\hat{a}-a)I_{\{\hat{a}=a^+\}}+E(\hat{a}-a)I_{\{\hat{a}=a^-\}}\nonumber\\
&=&\int_{\{\hat{a}=a^+\}}(\hat{a}-a) \mathrm{d}\mathbb{P}+\int_{\{\hat{a}=a^-\}}(\hat{a}-a) \mathrm{d}\mathbb{P}\nonumber\\
&=&\int_{\{\hat{a}=a^+\}}(a^+-a) \mathrm{d}\mathbb{P}+\int_{\{\hat{a}=a^-\}}(a^--a) \mathrm{d}\mathbb{P}\nonumber\\
&=&\left[\int_{\{\hat{a}=a^+\}}(a^+-a) \mathrm{d}\mathbb{P}+\int_{\{\hat{a}=a^-\}}(a^+-a) \mathrm{d}\mathbb{P} \right]\\
&& -\left[\int_{\{\hat{a}=a^-\}}(a^+-a) \mathrm{d}\mathbb{P}+\int_{\{\hat{a}=a^-\}}(a^--a) \mathrm{d}\mathbb{P} \right]\nonumber\\
&=&E\left(a^+-a\right)-\int_{\{\hat{a}=a^-\}}(a^+-a^-) \mathrm{d}\mathbb{P}\nonumber\\
&=&-E(a^+-a^-)I_{\{\hat{a}=a^-\}}\label{eqn-3}.
\end{eqnarray}
Here, equation (\ref{eqn-3}) is due to (\ref{eqn-1}). Recall that
\begin{eqnarray}
a^+-a^-&=&\frac{2\sum_{i<j}(X_i^r-X_j^r)(Y_i^r-Y_j^r)}{\sum_{i<j}(X_i^c-X_j^c)^2+\sum_{i<j}(X_i^r-X_j^r)^2}\nonumber\\
&=&\frac{2\sum_{i<j}(X_i^r-X_j^r)\left[a(X_i^r-X_j^r)+(\eta_i-\eta_j)\right]}
{\sum_{i<j}(X_i^c-X_j^c)^2+\sum_{i<j}(X_i^r-X_j^r)^2},\label{a+-a-}
\end{eqnarray}
since $a\geq0$. Therefore,
\begin{eqnarray}
E(\hat{a}-a)
&=&-E\left\{\frac{2\sum_{i<j}(X_i^r-X_j^r)\left[a(X_i^r-X_j^r)+(\eta_i-\eta_j)\right]}
{\sum_{i<j}(X_i^c-X_j^c)^2+\sum_{i<j}(X_i^r-X_j^r)^2}\right\}I_{\{\hat{a}=a^-\}}\nonumber\\
&=&-\frac{2\sum_{i<j}\left[|a|(X_i^r-X_j^r)^2P(\hat{a}=a^-)+(X_i^r-X_j^r)E(\eta_i-\eta_j)I_{\{\hat{a}=a^-\}}\right]}
{\sum_{i<j}(X_i^c-X_j^c)^2+\sum_{i<j}(X_i^r-X_j^r)^2}\nonumber\\
&=&-\frac{2\sum_{i<j}(X_i^r-X_j^r)^2P(\hat{a}=a^-)}{\sum_{i<j}(X_i^c-X_j^c)^2+\sum_{i<j}(X_i^r-X_j^r)^2}\nonumber\\
&=&-\frac{2aS^2(X^r)}{S^2(X^c)+S^2(X^r)}P(\hat{a}=a^-).\label{rst-1}
\end{eqnarray}
Similar to the preceding arguments,
\begin{eqnarray*}
E(|\hat{a}|-|a|)&=&E(|\hat{a}|-a)=E(|\hat{a}|-a)I_{\{\hat{a}=a^+\}}+E(|\hat{a}|-a)I_{\{\hat{a}=a^-\}}\\
&=&\int_{\{\hat{a}=a^+\}}(a^+-a) \mathrm{d}\mathbb{P}+\int_{\{\hat{a}=a^-\}}(-a^--a) \mathrm{d}\mathbb{P}\\
&=&E(a^+-a)-\int_{\{\hat{a}=a^-\}}(a^+-a) \mathrm{d}\mathbb{P}-\int_{\{\hat{a}=a^-\}}(-a^-+a) \mathrm{d}\mathbb{P}\\
&=&-E(a^++a^-)I{\{\hat{a}=a^-\}}.\\
\end{eqnarray*}
Recall again that
\begin{eqnarray}
a^++a^-%&=&\frac{2\sum_{i<j}(X_i^c-X_j^c)(Y_i^c-Y_j^c)}{\sum_{i<j}(X_i^c-X_j^c)^2+\sum_{i<j}(X_i^r-X_j^r)^2}\nonumber\\
&=&\frac{2\sum_{i<j}(X_i^c-X_j^c)\left[a(X_i^c-X_j^c)+(\lambda_i-\lambda_j)\right]}{S^2\left(X^c\right)+S^2\left(X^r\right)}.\label{a++a-}
\end{eqnarray}
It follows that
\begin{eqnarray}
E(|\hat{a}|-|a|)
%&=&-\frac{2}{S^2(X^c)+S^2(X^r)}\sum_{i<j}\left[a(X_i^c-X_j^c)P(\hat{a}=a^-)+(X_i^c-X_j^c)E(\lambda_i-\lambda_j|X)I_{\{\hat{a}=a^-\}}\right]\nonumber\\
&=&-\frac{2aS^2(X^c)}{S^2(X^c)+S^2(X^r)}P(\hat{a}=a^-).\label{rst-2}
\end{eqnarray}

\textbf{Case II: $a<0$}\\

In this case, we have
\begin{align*}
&a^+-a
=\frac{\sum_{i<j}(X_i^c-X_j^c)(\lambda_i-\lambda_j)+\sum_{i<j}(X_i^r-X_j^r)[-2a(X_i^r-X_j^r)+(\eta_i-\eta_j)]}{S^2(X^c)+S^2(X^r)},\\
&a^--a
=\frac{\sum_{i<j}(X_i^c-X_j^c)(\lambda_i-\lambda_j)-\sum_{i<j}(X_i^r-X_j^r)(\eta_i-\eta_j)}{S^2(X^c)+S^2(X^r)}.\\
\end{align*}
These imply
\begin{align*}
&E(a^+-a)=-\frac{2aS^2(X^r)}{S^2(X^c)+S^2(X^r)},\\
&E(a^--a)=0.\\
\end{align*}
Similar to the case of $a\geq 0$, we obtain 
\begin{align*}
&E(\hat{a}-a)=E(a^+-a^-)I_{\{\hat{a}=a^+\}},\\
&E(|\hat{a}|-|a|)=E(a^++a^-)I_{\{\hat{a}=a^+\}}.\\
\end{align*}
These, together with (\ref{a+-a-}) and (\ref{a++a-}), imply,
\begin{eqnarray}
&&E(\hat{a}-a)=-\frac{2aS^2(X^r)}{S^2(X^c)+S^2(X^r)}P(\hat{a}=a^+),\label{rst-3}\\
&&E(|\hat{a}|-|a|)=\frac{2aS^2(X^c)}{S^2(X^c)+S^2(X^r)}P(\hat{a}=a^+).\label{rst-4}
\end{eqnarray}
The desired result follows from (\ref{rst-1}), (\ref{rst-2}), (\ref{rst-3}) and (\ref{rst-4}).
\end{proof}

%========================================================================================================================%
\subsection{Proof of Theorem \ref{thm:ls_consist}}
\begin{proof}
From (\ref{b+}) and (\ref{b-}),
\begin{equation*}
E(\hat{b}|X)=E(\overline{Y^c}-\hat{a}\overline{X^c}|X)
=E(a\overline{X^c}+b+\overline{\lambda}-\hat{a}\overline{X^c}|X)
=\overline{X^c} E(a-\hat{a}|X)+b.
\end{equation*}
Similarly, from (\ref{mu+}) and (\ref{mu-}),
\begin{equation*}
E(\hat{\mu}|X)=E(\overline{Y^r}-|\hat{a}|\overline{X^r}|X)
=E(|a|\overline{X^r}+\overline{\eta}+\overline{\lambda}-|\hat{a}|\overline{X^r}|X)
=\overline{X^r} E(|a|-|\hat{a}|\big|X)+\mu.
\end{equation*}
Hence, the desired result follows by Proposition \ref{prop:ls_exp} and Lemma \ref{lem:sign-consist} in the Appendix.
\end{proof}

%===================================================================================================================%
\section{Appendix II: Lemmas}
%====================================================lemma 1========================================================%
\begin{lemma}\label{lem:cov_r}
Assume model (\ref{mod-1**})-(\ref{mod-2**}) and $\text{Var}\left(X^r\right)<\infty$. Then $\text{Cov}(X^r,Y^r) \geq 0$. Consequently, $S(X^r,Y^r) \geq 0 $ with probability converging to 1. 
\end{lemma}
\begin{proof}
According to (\ref{mod-2**}),
\begin{eqnarray}
  \text{Cov}\left(X^r,Y^r\right)
  &=& E\left(X^rY^r\right)-E\left(X^r\right)E\left(Y^r\right)\nonumber\\
  &=& E\left[X^r\left(|a|X^r+\eta_1\right)\right]-E\left(X^r\right)E\left(|a|X^r+\eta_1\right)\nonumber\\
  &=& |a|E\left(X^r\right)^2+\mu E\left(X^r\right)-|a|\left[E\left(X^r\right)\right]^2-\mu E\left(X^r\right)\nonumber\\
  &=& |a|\text{Var}\left(X^r\right)\nonumber\\
  &\geq& 0,\label{cov_true}
\end{eqnarray}
provided that $\text{Var}\left(X^r\right)<\infty$. %Separately, by Lemma \ref{lem:cov_est}, 
%\begin{equation*}
%  S\left(X^r,Y^r\right)=\frac{1}{n^2}\sum_{i<j}(X_i^r-X_j^r)(Y_i^r-Y_j^r)=\bar{X^rY^r}-\bar{X^r}\bar{Y^r}.
%\end{equation*}
By the SLLN, 
\begin{equation}\label{cov_sample}
  S\left(X^r,Y^r\right)\to\text{Cov}\left(X^r,Y^r\right)\ a.s..
\end{equation}
(\ref{cov_sample}) together with (\ref{cov_true}) completes the proof.
\end{proof}

%====================================================lemma 2========================================================%
\begin{lemma}\label{lem:cov_est}
The following are true for $v\in\left\{c,r\right\}$:
\begin{align}
S\left(X^v, Y^v\right)&=\frac{1}{n^2}\sum_{i<j}(X_i^v-X_j^v)(Y_i^v-Y_j^v),\label{est_cov}\\
S^2\left(X^v\right)&=\frac{1}{n^2}\sum_{i<j}(X_i^v-X_j^v)^2.\label{est_var}
\end{align}
\end{lemma}
\begin{proof}
To prove \eqref{est_cov},
\begin{align*}
&\sum_{i<j}(X_i^v-X_j^v)(Y_i^v-Y_j^v)=\sum_{i<j}(X_i^vY_i^v-X_i^vY_j^v-X_j^vY_i^v+X_j^vY_j^v)\\
&=\sum_{i<j}(X_i^vY_i^v+X_j^vY_j^v)-\sum_{i<j}(X_i^vY_j^v+X_j^vY_i^v)\\
&=(n-1)\sum_{i=1}^nX_i^vY_i^v-[(\sum_{i=1}^nX_i^v)(\sum_{i=1}^nY_i^v)-\sum_{i=1}^nX_i^vY_i^v]\\
&=n\sum_{i=1}^nX_i^vY_i^v-(\sum_{i=1}^nX_i^v)(\sum_{i=1}^nY_i^v)=n^2S\left(X^v, Y^v\right).\\
\end{align*}
\eqref{est_var} follows by replacing $Y_i^v$ with $X_i^v$ and $Y_i^v$ with $X_j^v$ in the above calculations.
\end{proof}

%====================================================lemma 3========================================================%
\begin{lemma}\label{lem:sign-consist}
Assume model (\ref{mod-1**})-(\ref{mod-2**}). Assume in addition that $S^2\left(X^c\right)=O(1)$ and $S^2\left(X^r\right)=O(1)$. Let $\left\{\hat{a}, \hat{b}, \hat{\mu}\right\}$ be the least squares solution defined in (\ref{def-ls}). Then
\begin{eqnarray*}
  && P\left(\hat{a}=a^{-}|a\geq 0\right)\to 0,\\
  && P\left(\hat{a}=a^{+}|a<0\right)\to 0,
\end{eqnarray*}
as $n\to\infty$.
\end{lemma}
\begin{proof}
We prove the case $a\geq 0$ only. The case $a<0$ can be proved similarly. Under the assumption that $a\geq 0$, 
\begin{equation*}
  \text{Cov}\left(X^c, Y^c\right)=a\text{Var}\left(X^c\right)\geq0,
\end{equation*}
and consequently, $P\left(S\left(X^c, Y^c\right)<0\right)\to 0$. According to Theorem \ref{thm:ls_solu}, the only other circumstance under which $\hat{a}=a^{-}$ is when $S\left(X^r, Y^r\right)>S\left(X^c, Y^c\right)>0$ and $L\left(a^+, b^+, \mu^+\right)>L\left(a^-, b^-, \mu^-\right)$ simultaneously. It is therefore sufficient to show that
\begin{eqnarray}
  && P\left(S\left(X^r, Y^r\right)>S\left(X^c, Y^c\right)>0, L\left(a^+, b^+, \mu^+\right)>L\left(a^-, b^-, \mu^-\right)\right)\label{eqn:goal}\\
  && \to 0\nonumber.
\end{eqnarray}
Notice
\begin{eqnarray*}
  && L\left(a^+, b^+, \mu^+\right)-L\left(a^-, b^-, \mu^-\right)\\
  =&& \frac{1}{n}\sum_{i=1}^{n}\left[\left(a^+X_i^c+b-Y_i^c\right)^2-\left(a^-X_i^c+b-Y_i^c\right)^2\right]\\
  &&+\frac{1}{n}\sum_{i=1}^{n}\left[\left(a^+X_i^r+\mu-Y_i^r\right)^2-\left(a^-X_i^r+\mu-Y_i^r\right)^2\right]\\
  :=&& \frac{1}{n}\left(I+II\right).
\end{eqnarray*}
The first term
\begin{eqnarray*}
  I&=& \sum_{i=1}^{n}\left[\left(a^+X_i^c+b-Y_i^c\right)^2-\left(a^-X_i^c+b-Y_i^c\right)^2\right]\\
  &=& \sum_{i=1}^{n}\left[\left(a^{+}-a\right)^2\left(X_i^c-\overline{X^c}\right)^2+\left(\lambda_i-\overline{\lambda}\right)^2
  -2\left(a^{+}-a\right)\left(X_i^c-\overline{X^c}\right)\left(\lambda_i-\overline{\lambda}\right)\right]\\
  && -\sum_{i=1}^{n}\left[\left(a^{-}-a\right)^2\left(X_i^c-\overline{X^c}\right)^2+\left(\lambda_i-\overline{\lambda}\right)^2
  -2\left(a^{-}-a\right)\left(X_i^c-\overline{X^c}\right)\left(\lambda_i-\overline{\lambda}\right)\right]\\
  &=& \left[\left(a^{+}-a\right)^2-\left(a^{-}-a\right)^2\right]\sum_{i=1}^{n}\left(X_i^c-\overline{X^c}\right)^2\\
  && -2\left(a^{+}-a^{-}\right)\sum_{i=1}^{n}
  \left(X_i^c-\overline{X^c}\right)\left(\lambda_i-\overline{\lambda}\right)\\
  &=& \left(a^{+}-a^{-}\right)\left[\left(a^{+}+a^{-}-2a\right)\sum_{i=1}^{n}\left(X_i^c-\overline{X^c}\right)^2
  -2\sum_{i=1}^{n}\left(X_i^c-\overline{X^c}\right)\left(\lambda_i-\overline{\lambda}\right)\right].
\end{eqnarray*}
From this, and the assumption that $S\left(X^r, Y^r\right)>S\left(X^c, Y^c\right)>0$, we see that $I>0$ is equivalent to 
\begin{eqnarray}
  && \left(\frac{a^{+}+a^{-}}{2}-a\right)\sum_{i=1}^{n}\left(X_i^c-\overline{X^c}\right)^2
  -\sum_{i=1}^{n}\left(X_i^c-\overline{X^c}\right)\left(\lambda_i-\overline{\lambda}\right)\label{eqn:consist-1}\\
  && >0\nonumber.
\end{eqnarray}
On the other hand, 
\begin{eqnarray*}
  &&(\ref{eqn:consist-1})\\
  &&= \left[\frac{S\left(X^c, Y^c\right)}{S^2\left(X^c\right)+S^2\left(X^r\right)}-a\right]\sum_{i=1}^{n}\left(X_i^c-\overline{X^c}\right)^2
  -\sum_{i=1}^{n}\left(X_i^c-\overline{X^c}\right)\left(\lambda_i-\overline{\lambda}\right)\\
  &&=\left[\frac{\sum_{i<j}\left(X_i^c-X_j^c\right)\left(\lambda_i-\lambda_j\right)}{\sum_{i<j}\left(X_i^c-X_j^c\right)^2+\sum_{i<j}\left(X_i^r-X_j^r\right)^2}
  -a\frac{S^2\left(X^r\right)}{S^2\left(X^c\right)+S^2\left(X^r\right)}\right]\sum_{i=1}^{n}\left(X_i^c-\overline{X^c}\right)^2\\
  &&-\sum_{i=1}^{n}\left(X_i^c-\overline{X^c}\right)\left(\lambda_i-\overline{\lambda}\right)\\
  &&= \frac{\sum_{i=1}^{n}\left(X_i^c-\overline{X^c}\right)^2}{\sum_{i<j}\left(X_i^c-X_j^c\right)^2+\sum_{i<j}\left(X_i^r-X_j^r\right)^2}
  \sum_{i<j}\left(X_i^c-X_j^c\right)\left(\lambda_i-\lambda_j\right)\\
  &&-\sum_{i=1}^{n}\left(X_i^c-\overline{X^c}\right)\left(\lambda_i-\overline{\lambda}\right)
  -a\frac{S^2\left(X^r\right)}{S^2\left(X^c\right)+S^2\left(X^r\right)}\sum_{i=1}^{n}\left(X_i^c-\overline{X^c}\right)^2\\
  &&= \frac{\sum_{i=1}^{n}\left(X_i^c-\overline{X^c}\right)^2}{\sum_{i<j}\left(X_i^c-X_j^c\right)^2+\sum_{i<j}\left(X_i^r-X_j^r\right)^2}
  \left[n\sum_{i=1}^{n}\left(X_i^c-\overline{X^c}\right)\left(\lambda_i-\overline{\lambda}\right)\right]\\
  &&-\sum_{i=1}^{n}\left(X_i^c-\overline{X^c}\right)\left(\lambda_i-\overline{\lambda}\right)
  -a\frac{S^2\left(X^r\right)}{S^2\left(X^c\right)+S^2\left(X^r\right)}\sum_{i=1}^{n}\left(X_i^c-\overline{X^c}\right)^2\\
  &&=\sum_{i=1}^{n}\left(X_i^c-\overline{X^c}\right)\left(\lambda_i-\overline{\lambda}\right)
  \left[\frac{S^2\left(X^c\right)}{S^2\left(X^c\right)+S^2\left(X^r\right)}-1\right]\\
  &&-a\frac{S^2\left(X^r\right)}{S^2\left(X^c\right)+S^2\left(X^r\right)}\sum_{i=1}^{n}\left(X_i^c-\overline{X^c}\right)^2\\
  &&=-\frac{S^2\left(X^r\right)}{S^2\left(X^c\right)+S^2\left(X^r\right)}
  n\left[aS^2\left(X^c\right)+S\left(X^c, \lambda\right)\right],
\end{eqnarray*}
where $S\left(X^c, \lambda\right)=\frac{1}{n}\sum_{i=1}^{n}\left(X_i^c-\overline{X^c}\right)\left(\lambda_i-\overline{\lambda}\right)$ denotes the sample covariance of the random variables $X^c$ and $\lambda$, which converges to $0$ almost surely by the independence assumption. Therefore,
\begin{eqnarray}
  \frac{1}{n}I&=&-2\left(a^+-a^-\right)\frac{S^2\left(X^r\right)}{S^2\left(X^c\right)+S^2\left(X^r\right)}\left[aS^2\left(X^c\right)+S\left(X^c, \lambda\right)\right]\nonumber\\
  &&\to C_1<0\label{eqn:consist-2}
\end{eqnarray}
almost surely, as $n\to\infty$. \\

By the similar calculation, we have that the second term
\begin{eqnarray}
  \frac{1}{n}II&=&-2\left(|a^+|-|a^-|\right)\frac{S^2\left(X^c\right)}{S^2\left(X^c\right)+S^2\left(X^r\right)}\left[aS^2\left(X^r\right)+S\left(X^r, \eta\right)\right]\nonumber\\
  &&\to C_2<0\label{eqn:consist-3}
\end{eqnarray}
almost surely, as $n\to\infty$. (\ref{eqn:consist-2}) and (\ref{eqn:consist-3}) together imply that 
\begin{equation*}
  P\left(\hat{a}=a^{-}|a\geq 0\right)\to 0.
\end{equation*}
This completes the proof.
\end{proof}

\end{document}